\documentclass[reprint,longbibliography,aps,prb,superscriptaddress]{revtex4-2}
\usepackage[english]{babel}
\usepackage{amsmath,amsthm,amssymb,bbm,amsfonts,color,graphicx,xcolor,framed}
\usepackage{url}
\usepackage{tikz}
\usepackage{enumerate}
\usepackage{mathrsfs}
\usepackage{relsize}
\usepackage[colorlinks=true, allcolors=blue]{hyperref}
\usepackage[utf8]{inputenc}
\usepackage{listings}
  \usepackage{stmaryrd}

\usepackage[braket, qm]{qcircuit}
\usepackage{graphicx}

\lstdefinestyle{mystyle}{
  backgroundcolor=\color{backcolour},
  commentstyle=\color{codegreen},
  keywordstyle=\color{magenta},
  numberstyle=\tiny\color{codegray},
  stringstyle=\color{codepurple},
  basicstyle=\ttfamily\footnotesize,
  breakatwhitespace=false,
  breaklines=true,
  captionpos=b,
  keepspaces=true,
  numbers=left,
  numbersep=5pt,
  showspaces=false,
  showstringspaces=false,
  showtabs=false,
  tabsize=2
}
\newcommand{\la}{\lambda}
\newcommand{\La}{\Lambda}
\newcommand{\al}{\alpha}

\newcommand{\si}{{\sigma}}

\newcommand{\HEIS}{{\mrm{HEIS}}}

\newcommand{\ii}{\mathrm{i}}
\newcommand{\ee}{\mathrm{e}}

\newcommand{\eq}[1]{Eq.~(\ref{eq:#1})}
\newcommand{\fig}[1]{Fig.~\ref{fig:#1}}
\renewcommand{\sec}[1]{Sec.~\ref{sec:#1}}

\newcommand{\app}[1]{Appendix~\ref{sec:#1}}
\newcommand{\tab}[1]{Table~\ref{tab:#1}}
\newcommand{\thm}[1]{Theorem~\ref{thm:#1}}
\newcommand{\alg}[1]{Algorithm~\ref{alg:#1}}

\newcommand{\SWAP}{\mathrm{SWAP}}
\newcommand{\CNOT}{\mathrm{CNOT}}

\newcommand{\id}{\mathbbm{1}}

\newcommand{\nn}{\nonumber}

\let\perptmp\perp
\renewcommand{\perp}{{\! \mathsmaller{\perptmp}}}

\newcommand{\mrm}{\mathrm}

\newtheorem{theorem}{Theorem}

\theoremstyle{definition}
\newtheorem{algorithm}{Algorithm}

\newcommand{\abox}[1]{\begin{array}{|c|} \hline #1 \\ \hline \end{array}}

\makeatletter
\def\@bibdataout@aps{%
  \immediate\write\@bibdataout{%
    @CONTROL{%
      apsrev41Control%
      \longbibliography@sw{%
        ,author="08",editor="1",pages="1",title="0",year="1"%
      }{%
        ,author="08",editor="1",pages="1",title="",year="1"%
      }%
    }%
  }%
  \if@filesw \immediate \write \@auxout {\string \citation {apsrev41Control}}\fi
}
\makeatother

\newcommand{\claw}{
  \hspace{-1em}
  \raisebox{-.25 em}{
    \begin{tikzpicture}[scale=0.1]
    \draw[gray, thick] (-1,1) -- (0,0);
    \draw[gray, thick] (1,1) -- (0,0);
    \draw[gray, thick] (0,-1.4142) -- (0,0);
    \filldraw[black] (-1,1) circle (5pt);
    \filldraw[black] (1,1) circle (5pt);
    \filldraw[black] (0,-1.4142) circle (5pt);
    \filldraw[black] (0,0) circle (5pt);
    \end{tikzpicture}
  }
  \hspace{-1em}
}

\newcommand{\heavy}{\mathrm{heavy}}

\definecolor{blue}{rgb}{0.12156862745098039, 0.4666666666666667, 0.7058823529411765}
\definecolor{orange}{rgb}{1.0, 0.4980392156862745, 0.054901960784313725}
\definecolor{green}{rgb}{0.17254901960784313, 0.6274509803921569, 0.17254901960784313}
\definecolor{red}{rgb}{0.8392156862745098, 0.15294117647058825, 0.1568627450980392}
\definecolor{4}{rgb}{0.5803921568627451, 0.403921568627451, 0.7411764705882353}
\definecolor{5}{rgb}{0.5490196078431373, 0.33725490196078434, 0.29411764705882354}
\definecolor{6}{rgb}{0.8901960784313725, 0.4666666666666667, 0.7607843137254902}
\definecolor{7}{rgb}{0.4980392156862745, 0.4980392156862745, 0.4980392156862745}
\definecolor{8}{rgb}{0.7372549019607844, 0.7411764705882353, 0.13333333333333333}
\definecolor{9}{rgb}{0.09019607843137255, 0.7450980392156863, 0.8117647058823529}

\begin{document}

\author{Joris Kattem\"olle}
\affiliation{Department of Physics, University of Konstanz, Konstanz, Germany}
\author{Seenivasan Hariharan}
\affiliation{Institute of Physics, University of Amsterdam, Amsterdam, The Netherlands}
\affiliation{
QuSoft, CWI, Amsterdam, The Netherlands}
\title{Line-graph qubit routing:\\from kagome to heavy-hex and more}

\begin{abstract}
  Quantum computers have the potential to outperform classical computers, but are currently limited in their capabilities. One such limitation is the restricted connectivity between qubits, as captured by the hardware's coupling graph. This limitation poses a challenge for running algorithms that require a coupling graph different from what the hardware can provide. To overcome this challenge and fully utilize the hardware, efficient qubit routing strategies are necessary. In this paper, we introduce line-graph qubit routing, a general method for routing qubits when the algorithm's coupling graph is a line graph and the hardware coupling graph is a heavy graph. Line-graph qubit routing is fast, deterministic, and effective; it requires a classical computational cost that scales at most quadratically with the number of gates in the original circuit, while producing a circuit with a SWAP overhead of at most two times the number of two-qubit gates in the original circuit. We implement line-graph qubit routing and demonstrate its effectiveness in mapping quantum circuits on kagome, checkerboard, and shuriken lattices to hardware with heavy-hex, heavy-square, and heavy-square-octagon coupling graphs, respectively. Benchmarking shows the ability of line-graph qubit routing to outperform established general-purpose methods, both in the required classical wall-clock time and in the quality of the solution that is found. Line-graph qubit routing has direct applications in the quantum simulation of lattice-based models and aids the exploration of the capabilities of near-term quantum hardware.
\end{abstract}

\maketitle

\section{Introduction}\label{sec:introduction}
Quantum computing offers potential to revolutionize a wide range of domains by efficiently solving problems that are intractable for classical computers~\cite{montanaro2016quantum,cerezo2021variational}. To run any quantum algorithm, it must be compiled into a quantum circuit that can be executed on the quantum hardware. The hardware coupling graph of a quantum computer, which defines adjacency between qubits based on the ability to perform two-qubit gates between them, plays a crucial role in this process. The problem of ensuring that a quantum circuit is compatible with the hardware coupling graph is referred to as the qubit routing problem~\cite{childs2019circuit, cowtan2019qubit}. While generalized methods exist for qubit routing~\cite{bapat2023advantages}, a standard approach to implement two-qubit gates between non-adjacent qubits is to insert SWAP gates, making the qubits effectively adjacent~\cite{maslov2008quantum,childs2019circuit,li2019mapping,siraichi2019routing,wille2019mapping,cowtan2019qubit}. To obtain a practical quantum advantage on noisy intermediate-scale quantum (NISQ)~\cite{preskill2018qcnisq} devices, it is imperative that overhead arising from compilation is kept to a minimum~\cite{franca2021limitations,Martiel2022architectureaware}. However, finding the  swapping strategy that requires the least number of SWAP gates is NP-hard~\cite{maslov2008quantum,childs2019circuit}, making it a challenging problem to solve. Heuristic, probabilistic methods have been developed, but their classical runtime may become problematic for large circuits, and the solution they find may be far from optimal~\cite{li2019mapping,siraichi2019routing,wille2019mapping,cowtan2019qubit}. Striking the right balance between the classical resources required for routing and minimizing the circuit depth of the routed quantum circuit is crucial in maximizing the performance of NISQ devices.

The qubit routing problem is particularly evident in the quantum simulation of lattice-based spin models. One of the first areas in which it was realized that quantum computers could outperform classical computers was that of quantum simulation~\cite{feynman1982simulating}. When applied specifically to lattice-based spin systems with two-body interactions, quantum simulation by Trotterization~\cite{lloyd1996universal} approximates the overall time evolution operator of the quantum-mechanical system by a sequence of two-qubit gates, where each two-qubit gate corresponds to the time evolution according to one two-body term in the Hamiltonian (dynamic quantum simulation)~\cite{burkard2022recipes,kattemolle2022variational}. Additionally, by introducing variable parameters for the per-term evolution times, these circuits are transformed to circuits that prepare ansatz states for the variational quantum eigensolver (VQE)~\cite{wecker2015progress}, designed to variationally find the ground state of the quantum-mechanical system (static quantum simulation). Before any routing, these circuits for dynamic and static quantum simulation naturally require hardware with a coupling graph that is equal to the lattice of the lattice-based spin system (the virtual graph)~\cite{wecker2015progress}. There will generally be a mismatch between the virtual graph and the hardware coupling graph. Efficient qubit routing plays a crucial role in overcoming this mismatch.

\begin{table*}
\newcommand{\myfig}[1]{\parbox{10.3em}{\hspace{.5em}\includegraphics[width=9.3em]{#1.pdf}\hspace{.5em}\vspace{1em}}}
\def\arraystretch{2}
\setlength\tabcolsep{.5em}
\begin{tabular}{c|c|c|c|c}
& \myfig{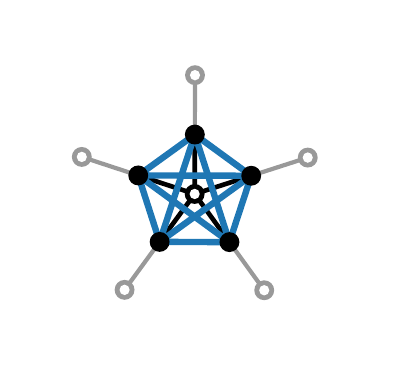}            & \myfig{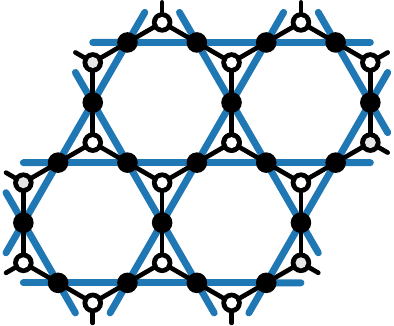}              
& \myfig{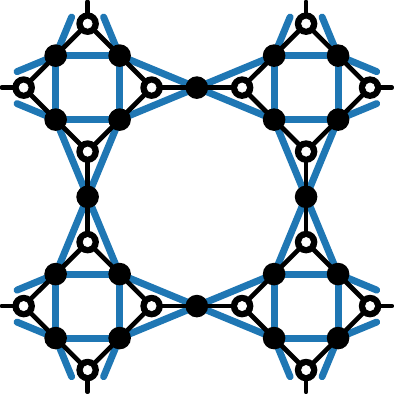}      & \myfig{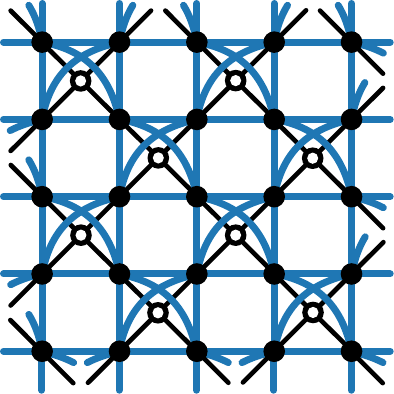}  
\\ \hline
$L(G)$                  & complete                    & kagome                      
& \parbox{10em}{shuriken/ \\ square-kagome \vspace{.3em}}              & checkerboard          \\
$\heavy(G)$             & {\color{gray}(heavy)-}star  & heavy-hex                   
& heavy-square-octagon 
& \parbox{10em}{heavy-square/\\Lieb lattice~\cite{Lieb1989heavysquare}\vspace{.3em}}          
\\ \hline
(model) material                
& spin glasses~\cite{Sherrington1975spinglass, Parisi1979spinglass}          
& herbertsmithite~\cite{norman2016herbertsmithite}     
& atlasovite-like~\cite{Siddharthan2001squarekagomeHAFM, Fujihala2020squarekagome}      & planar pyrochlore~\cite{singh1998paradigm,fouet2003planar,bishop2012frustrated}  \\
\end{tabular}
  \caption{\label{tab:examples} Examples of line-graph routing. Line-graph routing maps any circuit with coupling graph $L(G)$ (blue edges) to circuits with coupling graph $\heavy(G)$ (black edges). Line-graph routing finds direct application in the quantum simulation of the magnetic properties of some (model) materials with coupling graph $L(G)$ (last row) on hardware with coupling graph $\heavy(G)$.}
\end{table*}

In this paper, we develop an efficient and deterministic qubit routing strategy, which we call line-graph qubit routing, or line-graph routing for short. It maps any circuit on a line graph $L(G)$ to hardware with coupling graph $\heavy(G)$. Here, $\heavy(G)$ is obtained from the graph $G$ by placing a node on every edge of $G$. We call these added nodes the heavy nodes of $\heavy(G)$. By definition, the nodes of the line graph $L(G)$ consist of the heavy nodes of $\heavy(G)$. In $L(G)$, two nodes are adjacent if the associated edges in $G$ are incident on the same node of $G$. It is instructive to verify this property for one of the pairs $(L(G),\heavy(G))$ in \tab{examples}. For example, line-graph routing maps the circuits for the quantum simulation of the Heisenberg anti-ferromagnet (HAFM) on the kagome lattice to hardware with a heavy-hex coupling graph. In this example, $G$ is the hexagonal lattice and $L(G)$ is the kagome lattice. Despite these examples, we stress that line-graph routing is applicable to any circuit on any line graph~$L(G)$.

The remainder of this paper is organized as follows. We first introduce line-graph routing by example, mapping circuits for the quantum simulation of the kagome HAFM to hardware with a heavy-hex coupling graph (\sec{kagome_heavyhex}). We formalize and generalize this approach to arbitrary circuits and arbitrary line graphs in \sec{general_case}. We benchmark our software implementation of the general algorithm against existing qubit routing approaches in \sec{implementation}, to conclude with a discussion and outlook in \sec{discussion}.

\section{Kagome to heavy-hex}\label{sec:kagome_heavyhex}
Line-graph routing is arguably most clearly explained with an example, which we do in this section by mapping circuits for the quantum simulation of the kagome HAFM to quantum hardware with a heavy-hex coupling graph. First, we use this example due to the relevance of the kagome spin model in exploring quantum phenomena like topological states of matter and quantum spin liquids~\cite{norman2016herbertsmithite,savary2017qsl}. The kagome lattice's significance extends to chemistry, as it is frequently observed in transition metal compounds and metal organic frameworks~\cite{nocera2004inorganic, chakraborty2021mof}. The ground state of the kagome HAFM is a long-standing open problem in quantum magnetism~\cite{lauchli2019kagomehafm} that can potentially be solved on NISQ devices~\cite{kattemolle2022variational}. By classical emulation of noiseless quantum computers, it was previously demonstrated that the ground-state energy found by a VQE approaches the true ground-state energy exponentially as a function of the circuit depth~\cite{kattemolle2022variational}.

Second, the heavy-hex coupling graph is the coupling graph of IBM's current and future superconducting hardware~\cite{nation_paik_cross_nazario_2022,bravyi2022quantumcentric}. Among the emerging quantum hardware platforms, IBM's superconducting qubits have gained significant attention due to their rapid development and scalability in the NISQ era. To optimize this superconducting qubit hardware and mitigate the occurrence of frequency collisions and crosstalk~\cite{brink2018frequencycollisions,hertzberg2021frequencycollisions, Sarovar2020detectingcrosstalk}, error correcting codes are designed on low-degree graphs such as heavy-hex and heavy-square lattices, preventing errors during program execution~\cite{chamberland2020heavylattices,nation_paik_cross_nazario_2022, bravyi2022quantumcentric}. This motivates the further development of hardware with these types of connectivity graphs.

The relevance of the kagome-to-heavy-hex mapping was further highlighted by the IBM Quantum’s Open Science Prize 2022, where the challenge was to prepare the ground state of the kagome HAFM using a VQE and implement it on a 16-qubit IBM Quantum Falcon device with a heavy-hex coupling graph~\cite{Lanes_Rasmusson_2023IBMprize}. It is important to note that, also within the context of quantum simulation, the routing problem is not unique to the quantum simulation of spin problems on the kagome lattice. Other lattice-based spin models, such the HAFM on the shuriken lattice, are also known for their geometric frustration and challenging simulation~\cite{Wu2023vscore}. 

\subsection{Line-graph routing}\label{sec:line_graph_routing}
The first step in line-graph routing is establishing a one-to-one correspondence between the nodes of the virtual graph (in this section, the kagome lattice) and the hardware connectivity graph (in this section, the heavy-hex lattice). This correspondence is achieved by aligning the nodes of the kagome lattice with the heavy nodes of the heavy-hex lattice, as shown in \fig{four_coloring}. Subsequently, the light (non-heavy) nodes of the heavy-hex lattice are used to mediate two-qubit gates between the spins on the nodes of the kagome lattice.

To see this in more detail, assume a kagome quantum circuit, that is, a circuit composed of single-qubit gates on qubits $\{i\}$ and two-qubit gates along the edges $\{(i,j)\}$ of a patch of the kagome lattice. To map the circuit from the kagome to the heavy-hex lattice, we label the heavy qubits on the heavy-hex lattice with the labels $\{i\}$ of the congruent qubits on the kagome lattice, as shown in \fig{four_coloring}. Under this identification, any single-qubit gate in the kagome circuit is trivially mapped to a single-qubit gate on the heavy-hex lattice.

To map the two-qubit gates, let us label the $\ell$th two-qubit gate in the kagome circuit, acting on qubits $(i,j)$, by $U^{\ell}_{ij}$. Any such gate can be performed on the heavy-hex lattice by mapping it to the three-qubit \emph{mediated two-qubit gate} $MU$
\begin{equation}\label{eq:mediated_gate}
    U^\ell_{ij}\mapsto M\!U^{\ell}_{imj}=\SWAP_{mi} U^\ell_{mj} \SWAP_{i m},
\end{equation}
 where qubit $m=m_{ij}$ mediates the interaction between qubits $i$ and $j$. This map provides the cornerstone of line-graph routing. The key point of line-graph routing is that for every pair of qubits $(i,j)$ the existence and uniqueness of the mediating qubit $m=m_{ij}$ is guaranteed by the definition of $L(G)$ and $\heavy(G)$ (see \sec{introduction}).

Equation \eqref{eq:mediated_gate} introduces many SWAP gates that need not be performed physically. First, SWAP gates occurring at the beginning and end of the routed circuit can be accounted for by a relabeling of the qubits. Second, any two consecutive SWAP gates can be cancelled. These double SWAP gates are introduced by \eq{mediated_gate} if there are two consecutive two-qubit gates acting on the same two qubits (possibly with additional single-qubit gates on those qubits in between). Double SWAP gates are also introduced by \eq{mediated_gate} in the case of two consecutive two-qubit gates that have a single qubit in common and where the two resulting mediated gates have a mediating qubit in common. That is, if $m=m_{ij}=m_{ik}$, we have by \eq{mediated_gate} that
  \begin{align}\label{eq:cancel_SWAP}
    U^{\ell '}_{i k}U^{\ell}_{ij}&\!\mapsto\!(\SWAP_{mi}U^{\ell'}_{mk}\SWAP_{im})(\SWAP_{m i}U^\ell_{mj} \SWAP_{i m})\nn\\
    &\!=\SWAP_{mi}U^{\ell'}_{mk}U^\ell_{m j} \SWAP_{i m}.
  \end{align}

In summary, line-graph routing first associates the nodes of $L(G)$ with the heavy nodes of $\heavy(G)$, applies the map of \eq{mediated_gate} to all two-qubit gates, and finally removes superfluous SWAP gates as described above.

\subsection{Application: quantum simulation}\label{sec:quant_sim}

\begin{figure}[t]
  \centering
    \includegraphics[width=.9 \columnwidth]{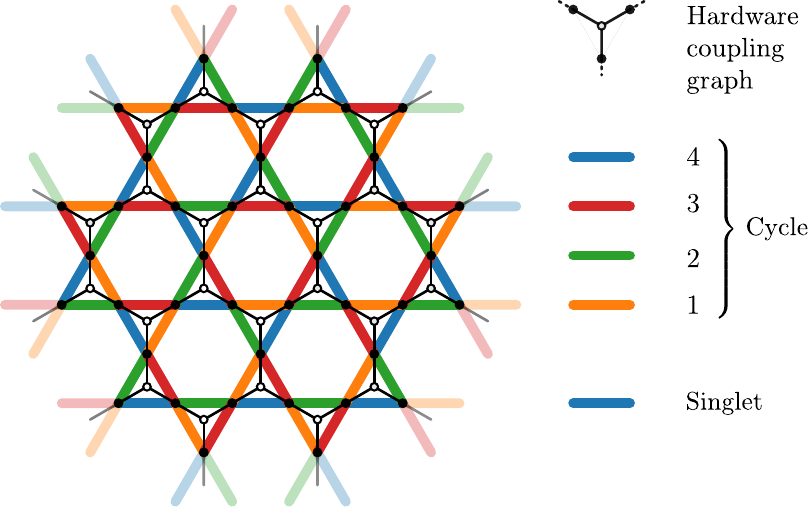}
    \caption{\label{fig:four_coloring} The kagome lattice (colored edges) is the line graph of the hexagonal lattice (black edges). As shown in the figure, the nodes of the kagome lattice can be identified with the heavy nodes of the heavy-hex lattice. The colored edges represent one possible circuit with a kagome coupling graph. Every color represents a layer in this circuit. In the first layer, singlet states are created along the blue lines. Thereafter, $\HEIS$ gates [\eq{HEIS_gate}] are applied along all colored edges in sequence, defining one circuit cycle. This cycle is repeated to obtain the complete circuit. Figure adapted from Ref.~\cite{kattemolle2022variational}.}
  \end{figure}
  
  One immediate application of the kagome-to-heavy-hex mapping is in the quantum simulation of the kagome HAFM on heavy-hex quantum hardware. In units where $\hbar=1$, the HAFM has a Hamiltonian
  \begin{align}\label{eq:heisenberg}
    H=\sum_{(i,j)} H_{ij}, &&H_{ij}=X_iX_j+Y_iY_j+Z_iZ_j,
  \end{align}
  where the sum is over all edges $(i,j)$ of a given graph. In the current section, this graph could be any patch of the kagome lattice. Here, $X_i$ denotes the Pauli-$X$ operator acting on qubit $i$ (similarly for $Y_i,Z_i$). This Hamiltonian is straightforwardly generalized to arbitrary two-spin interactions along the edges of a graph~\footnote{These are described by $H=\sum_{(i,j)} P_{ij}+\sum_i P_i$, with two-spin terms $P_{ij}=\sum_{kl}\gamma_{ij}^{(kl)}\si_i^{(k)}\si_j^{(l)}$ and single-spin terms $P_{i}=\sum_{k}\gamma_{i}^{(k)}\si_i^{(k)}$, where $\{\si^{(k)}_i\}=\{\id,X_i,Y_i,Z_i\}$.}, which would conceptually not change the constructions that follows.

\begin{figure*}
  \centering
  \includegraphics[width=\textwidth]{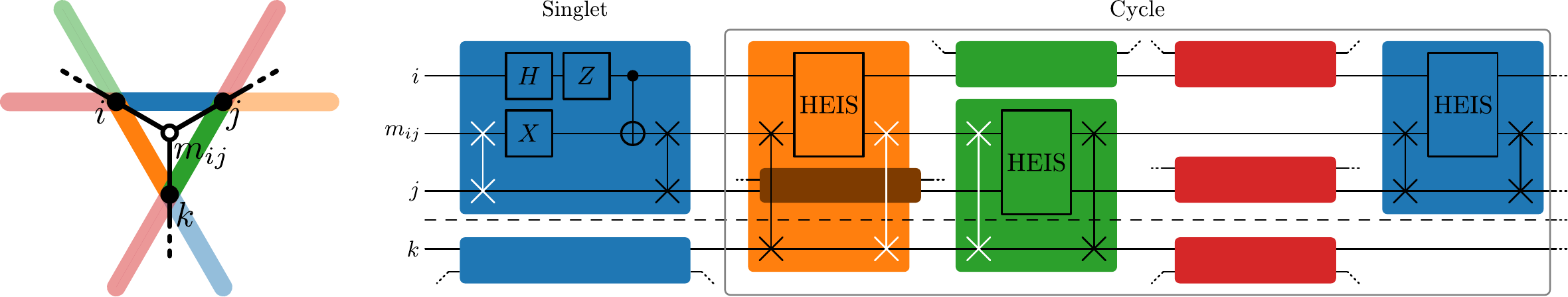}
  \caption{\label{fig:circuit} The line-graph routed quantum circuit for the quantum simulation of the kagome HAFM (\fig{four_coloring}), focusing on one triangle. The circuits involving the other triangles are similar. Empty colored boxes indicate subcircuits coming in from other triangles. The white $\SWAP$ gates can be omitted by a relabeling of the qubits (first SWAP gate) and the cancellation of double SWAP gates [other white SWAP gates, cf. \eq{cancel_SWAP}].}
\end{figure*}

  The goal of dynamic quantum simulation is to compute expectation values of observables with respect to the time-evolved state $\ket{\psi(t)}=\ee^{-\ii H t}\ket{\psi(0)}$, given some initial state $\ket{\psi(0)}$. On a quantum computer, this can be achieved by applying the unitary $\prod_{ij}\ee^{-\ii H_{ij} (t/r)}$ to $\ket{\psi(0)}$ (the latter of which is assumed to be easy to prepare) a total of $r$ times. That is,
  \begin{equation}\label{eq:trotterization}
    \ket{\psi(t)}\approx \left( \prod_{ij}\ee^{-\ii \frac{t}{r} H_{ij}}\right)^r\ket{\psi(0)}.
  \end{equation}
  The error in this approximation is of the order $t^2/r$~\cite{lloyd1996universal}, assuming a perfect quantum computer. After the preparation of $\ket{\psi(t)}$, expectation values can be extracted by repeated preparation and measurements. Note that $H_{ij}$ acts on two qubits, so that $\ee^{-\ii H_{ij} (t/r)}$ is a two-qubit unitary that can be decomposed into a few one and two-qubit gates that act on qubits $i,j$ only. In the case of the HAFM, $\ee^{-\ii H_{ij} (t/r)}$ is called the $\HEIS$ gate~\cite{kattemolle2022variational}
  	\begin{align}
		\mrm{HEIS}_{ij}(\al)&\equiv\ee^{-\ii\al/4}\ee^{-\ii \al H_{ij}/4},\label{eq:HEIS_gate}
    \end{align}
    where, in anticipation of static simulation, we have set $\al=4 t/r$, and where the physically irrelevant prefactor $\ee^{-\ii\al /4}$ and a factor of 1/4 in the exponent are included for consistency with Ref.~\cite{kattemolle2022variational}.

    We can go from the circuit for dynamical simulation [\eq{trotterization}] to the circuits needed for static quantum simulation by a VQE~\cite{peruzzo2014vqe, mcClean2016vqe, cerezo2021variational} by considering $\al$ as a free parameter in every occurrence of the HEIS gate. VQEs form are a promising method for computing ground state energies of various many-body systems in the NISQ era~\cite{peruzzo2014vqe, mcClean2016vqe, cerezo2021variational}. In VQEs, a parameterized quantum state is prepared with a parameterized circuit. The energy of this state is measured and optimized using a classical heuristic optimization routine. By the variational principle, the lowest energy that is found in this way provides an upper bound on the ground-state energy. In principle, there are no a priori restrictions on the structure of the ansatz circuit~\cite{peruzzo2014vqe}. In the subclass of VQEs using the so-called Hamiltonian variational ansatz (HVA), however, each gate in the parameterized circuit is either formed by parameterized time evolution along a term in the Hamiltonian, or by parameterized time evolution according to some reference Hamiltonian~\cite{wecker2015progress}. The initial state of the circuit is the (known and easy-to-prepare) ground state of the reference Hamiltonian.

    One possible quantum circuit following from the above considerations is depicted in \fig{four_coloring} (colored edges). This circuit can both be used for the dynamic quantum simulation (fixed parameters $\al$) or for static quantum simulation using the HVA (free parameters $\al$, to be optimized by a classical optimization routine). Here, every color depicts a different layer of the circuit. In the first layer, singlet states $(\ket{01}-\ket{10})/\sqrt{2}$ are placed along the blue edges. This provides the initial state of either the circuit for dynamical or static quantum simulation. It is the ground state of a reference Hamiltonian $\sum_{(i,j)} H_{ij}$, where the sum is over the blue edges $(i,j)$. After preparation of the initial state, $\HEIS$ gates are placed along all orange, green, red, and blue edges. This combination of $\HEIS$ gates defines one \emph{cycle} of the circuit. The cycle is repeated $r$ times. (The color blue defines both the initial state and the last layer of the cycle.) Alternative circuits for dynamic or static quantum simulation (using the HVA) may be achieved by changing the initial state or the order of the $\HEIS$ gates.

    Line-graph routing maps the kagome circuit from \fig{four_coloring} to a circuit on the heavy-hex lattice by first identifying nodes of the respective lattices, as in \fig{four_coloring}. Subsequently, \eq{mediated_gate} is applied to all gates, after which superfluous $\SWAP$ gates are removed [\eq{cancel_SWAP}]. The routed circuit thus obtained is shown in \fig{circuit}.

\section{General case}\label{sec:general_case}
Here, we formalize and generalize the routing strategy from the previous section, which leads to our main results, \thm{line-graph_routing} and \alg{line-graph_routing}. To introduce notation, we now define line- and heavy graphs more formally. The graph $\heavy(G)$ is obtained from $G$ by defining a node $i$ and edges  $(m,i)$, $(i,m')$ for each edge $(m,m')$ in $G$. The nodes $i$ are referred to as the heavy nodes. A minor but subtle point is that this construction may lead to paths of length three as induced subgraphs of $\heavy(G)$ in which two adjacent nodes are of degree one. That is, induced subgraphs of the form $P=\{(m,i),(i,m')\}$, with $m$ and $i$ nodes of degree one. For example, in the star graph (\tab{examples}), all paths emanating from the center node (i.e., the `rays' of the star) are of the form of $P$. Line-graph routing [previewed in \eq{mediated_gate}] will never use mediating qubits associated with nodes locally equivalent (including neighbors and next-nearest neighbors) to node $m$ in $P$ and can therefore be discarded without affecting the routed circuit. In the example of the star graph (Table \ref{tab:examples}), this means the qubits associated with the gray nodes may be removed from the routed circuit. To avoid frequent mention of minor but subtle point, in this paper we also refer to heavy graphs where the nodes locally equivalent (including neighbors and next-nearest neighbors) to node $m$ in $P$ are removed as heavy graphs. After these removals, we still refer to nodes locally equivalent (only including neighbors) to node $i$ in $P$ as heavy nodes.

The line graph of $G$, $L(G)$, is defined in terms of $\heavy(G)$ as follows: construct $\heavy(G)$ from $G$ and let the node set of $L(G)$ be equal to the set of heavy nodes $\{i\}$ of $\heavy(G)$. An edge $(i,j)$ is added to $L(G)$ if and only if there exists a node $m$ in $G$ such that $(i,m)$ and $(m,j)$ are edges in $\heavy(G)$.

We say that a quantum circuit $C$, consisting of one- and two-qubit gates by assumption, has coupling graph $G = (V,E)$ if $V$ is equal to the set of qubit labels in $C$ and if $(i,j)$ in $E$ if and only if there exists a two-qubit gate $U_{ij}$ in $C$. If a \emph{quantum computer} has coupling graph $G = (V,E)$, then it can perform arbitrary single-qubit gates on the qubits associated with each node in $V$ and arbitrary two-qubit gates along each edge in $E$. By \eq{mediated_gate}, we then have the following theorem.
\begin{theorem}[Line-graph routing]\label{thm:line-graph_routing}
  Every quantum circuit $C$ with coupling graph $L(G)$ can be performed on quantum hardware with coupling graph $\heavy(G)$ with a SWAP overhead of at most two times the number of two-qubit gates in $C$.
\end{theorem}
\begin{proof}
  By definition, there is a one-to-one correspondence between the nodes of $L(G)$ and the heavy nodes $\{i\}$ of $\heavy(G)$. Therefore, the single-qubit gates of $C$ can be mapped directly to hardware with coupling graph $\heavy(G)$. Furthermore, for every edge $(i,j)$ in $L(G)$, there are edges $(i,m)$ and $(m,j)$ in $\heavy(G)$, where $m=m_{ij}$ can be determined uniquely from $i$ and $j$. Thus, every two-qubit gate $U_{ij}$ in $C$ can be mapped to the three-qubit gate $MU_{imj}:=\SWAP_{mi}U_{mj}\SWAP_{im}$, leading to an overhead of 2$\la$ SWAP gates, with $\la$ the total number of two-qubit gates in $C$.
\end{proof}
Note that the above theorem does not make any assumption about the graph $G$. Also note that the theorem provides a  hierarchy of mappings, as $\heavy(G)$  itself may be a line graph. For example, invoking the above theorem twice gives a routing from $L(L(G))$ to $\heavy(L(G))$.

In practice, one may be given hardware with a coupling graph $H'$ and asked to find the class of circuits that can be run on this hardware using the line-graph construction. (For an overview of the graphs that follow, please see \fig{steps_algorithm}.) This task may arise if one has specific but limited quantum hardware available and wants to explore its capabilities by looking at quantum circuits that can be routed to this hardware with low overhead. The most immediate class of such quantum circuits is obviously formed by circuits with coupling graph $H'$. In case $H'$ is a heavy graph, $H'=\heavy(G)$, line-graph routing extends the possibilities to the class of circuits with coupling graph $L(G)$. If $H'$ is a heavy graph, it is trivial to find $G$ such that $H'=\heavy(G)$. It is also trivial to construct the line graph $L(G)$ from $G$. Therefore, given a heavy hardware coupling graph, it is trivial to find the class of circuits that can be run on it by line-graph routing. For example, given hardware with a heavy-hex coupling graph, it is trivial to find that line-graph routing yields the class of kagome circuits.

\begin{figure}
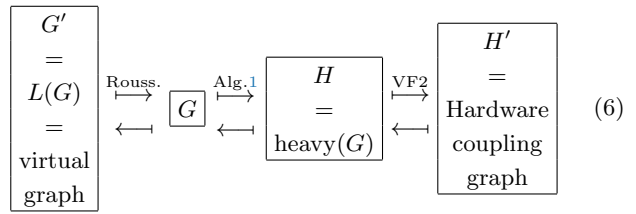

   \begin{equation}
  \abox{G'\\=\\L(G)\\=\\\mrm{virtual}\\\mrm{graph}}  \begin{array}{c}\overset{\mrm{Rouss.}}{\longmapsto}\\\longmapsfrom\end{array}
  \abox{G}
  \begin{array}{c}\overset{\mrm{Alg. \ref{alg:congruent_heavy_labels}}} {\longmapsto}\\\longmapsfrom\end{array}
  \abox{H\\=\\\heavy(G)}
  \begin{array}{c}\overset{\mrm{VF2}}{\longmapsto}\\\longmapsfrom\end{array} \abox{H'\\=\\\text{Hardware}\\\text{coupling}\\\text{graph}}
\end{equation}
  \caption{\label{fig:steps_algorithm} An overview of the graphs related to line-graph routing. All maps from the right to the left are trivial given that $H'$ is a heavy graph.}
\end{figure}

The task can also be reversed: given a circuit with a coupling graph $G'$, find the hardware on which it can be run with low overhead (cf. \fig{steps_algorithm}). Again, the first place to look would be hardware with coupling graph $G'$. Line-graph routing extends the possibilities by adding hardware with a coupling graph of $\heavy(G)$ with $G$ such that $G' = L(G)$.

But how to find $G$ from the circuit coupling graph $G'$? This is less straightforward, first because it is not possible to find $G$ from $G'$ if $G'$ is not a line graph. In this case, line-graph routing cannot be applied. There are numerous straightforward ways of checking whether a graph is a line graph. For example, Beineke's theorem states that a graph is a line graph if and only if it does not contain an induced subgraph out of a set of nine forbidden subgraphs~\cite{beineke1970derivedgraphs}. One of these forbidden graphs is the claw~\hbox{(\claw)}. For example, since the heavy-hex and hexagonal lattices consist entirely of claws, they are themselves not line graphs. Second, even if $G'$ is a line graph, it is nontrivial to find the graph $G$ such that $G'=L(G)$. Nevertheless, Roussopoulos' algorithm~\cite{roussopoulos1973graph} finds $G$ from $G'$, or reports that $G'$ is not a line graph, in time $O(\max\{n,\lvert E_{G'} \rvert\})$, with $n$ the number of nodes and $\lvert E_{G'} \rvert$ the number of edges of $G'$.

We briefly introduce the concepts from Roussopoulos' algorithm that are useful to us later. If $G'$ is a line graph, Roussopoulos' algorithm partitions the edges of $G'$ into complete subgraphs in such a way that no node lies in more than two of the subgraphs (which is possible if and only if $G'$ is a line graph). Then, the nodes of $G$ correspond to the sets in the partition. Additionally, the nodes that lie in only one of the sets in the partition are added as nodes of $G$ as sets of length one. Two nodes in $G$ are adjacent if their corresponding sets have a nonempty overlap. For example, in \fig{four_coloring}, each triangle of colored edges forms a set in the partition of the kagome lattice because triangles are fully connected subgraphs of the kagome lattice and no node of the kagome lattice is in more than two triangles. Defining the partitions as the nodes of $L(G)$ and putting an edge between two nodes in $L(G)$ whenever the corresponding sets have a nonempty intersection results in the hexagonal lattice (open nodes in \fig{four_coloring}). 

For line-graph routing [as previewed in \eq{mediated_gate}], $L(G)$ must be mapped to $\heavy(G)$ in such a way that the labels of the heavy nodes of $\heavy(G)$ are identical to the labels of the nodes of $L(G)$. The output of Roussopoulos' algorithm provides a convenient way to achieve this labeling through the following algorithm. It takes a graph $G$, as generated by Roussopoulos' algorithm, as input and returns $H=\heavy(G)$, a heavy graph where the set of heavy nodes is equal to the set of nodes (i.e., with equal labels) of $G'=L(G)$. This generalizes and automatizes the identification of nodes made in \fig{four_coloring}. 
\begin{algorithm}[Congruent heavy labels]\label{alg:congruent_heavy_labels}
Consider the edges $(a,b)$ of $G$ for which both $a$ and $b$ contain more than one node. (The nodes in $a$ and $b$ are from $G'$, the latter of which need not be provided explicitly to the current algorithm.) For all such edges $(a,b)$, add $(a,c)$ and $(c,b)$ to an empty graph $H$. Because $a$ and $c$ are distinct sets in a partition of the edges of $G'$ into complete subgraphs and $a$ and $c$ have a nonempty overlap, $c$ contains a single entry and this entry is a node from $G'$. Now, consider the edge cases, that is, the edges ${(a,b)}$ in $G$ where either $a$ or $b$ is a set of length one. For every such edge, add $(a,b)$ to $H$.
\end{algorithm}

Everything is now in place for the succinct presentation of line-graph routing. It maps any circuit $C$, with unknown coupling graph $G'$, to a circuit with coupling graph $H$ whenever $G'$ is a line graph, $G'=L(G)$, and where $H=\heavy(G)$.
\begin{algorithm}[Line-graph routing]\label{alg:line-graph_routing} 
Construct the coupling graph $G'$ of the circuit $C$. Run Roussopoulos' algorithm on $G'$ to obtain $G=L^{-1}(G')$. Construct $\heavy(G)$ using \alg{congruent_heavy_labels}. For every two-qubit gate $U_{ij}$ in $C$, let $m=m_{ij}$ be the node in $\heavy(G)$ that is in between nodes $i,j$ of $\heavy(G)$ and replace $U_{ij}$ by ${MU}_{imj}:=\SWAP_{mj}U_{im}\SWAP_{jm}$.
\end{algorithm}
Note the existence and uniqueness of mediating qubit $m$ are guaranteed by the definition of line graphs. In some cases, some qubits, related to the so-called lone leaf nodes in the coupling graph of the output circuit of line-graph routing, can be removed from that circuit, reducing the SWAP and qubit count of line-graph routed circuits. This removal leads to a marginal improvement, which may nevertheless be crucial given hardware with few qubits, but arguably obfuscates the general idea of line-graph routing. The definition of lone leaf nodes and the method for their removal is presented in \app{lone_leafs}.

As opposed to heuristic methods~\cite{li2019mapping}, line-graph routing is deterministic, allowing rigorous performance guarantees. Given the time complexity of Roussopoulos' algorithm~\cite{roussopoulos1973graph}, it is straightforward to show that the time complexity of line-graph routing is $O(\La^2)$, with $\La$ the number of gates in the input circuit $C$. A tighter bound on the time complexity can possibly obtained, but this requires a more sophisticated analysis. Such an analysis is unnecessary for the current purposes because a nonoptimized implementation of line-graph routing already routes circuits with thousands of qubits and hundred thousands of gates within a minute on a standard laptop~\cite{Note2}.

To utilize line-graph routing in practice, one additional step may be required. Quantum hardware providers will generally use a specific labeling of the qubits on their hardware, leading to a hardware coupling graph $H'$. This labeling may differ from the labeling of $H=\heavy(G)$ obtained through \alg{congruent_heavy_labels}. However, $H$ can be  mapped to a subgraph of $H'$ using an algorithm that finds subgraph isomorphisms, such as the VF2 algorithm~\cite{cordella2004VF2}. The VF2 algorithm generates a list of isomorphic subgraphs, and in practice, a selection is made based on a performance metric, such as the average two-qubit gate fidelity, to identify the subgraph with the best performance. This step is commonly implemented in quantum computing software development kits. In this paper, we also refer to $H$ as the hardware coupling graph and in this wording leave the possible mapping to the fixed qubit labels provided by a hardware provider implicit. 

\section{Implementation and benchmarking}\label{sec:implementation}

\setlength\tabcolsep{1em}
\begin{table*}
\begin{tabular}{ccccccc}
$L(G)$                    & Routing method  & Opt. depth  & $n_\SWAP$  & $n_\mrm{qubit}$  & $t_\mrm{tot}$ (s) & $\bar t$ (s)   \\
  \hline \hline
Complete                  & line-graph      & 1036        & 650       & 10               &   $<$1           & $<$1         \\
                          & SABRE           & 720         & 146       & 9                &   3              & $<$1         \\
\hline
Kagome                    & line-graph      & 226         & 7968      & 300              &   43             & 43           \\
                          & SABRE           & 790         & 8486      & 200              &  886             & 55 $\pm$ 1   \\
\hline
Shuriken                  & line-graph      & 209         & 13600     & 476              &  77              & 77           \\
                          & SABRE           & 1099        & 17064     & 323              &  2520            & 157 $\pm$ 4  \\
\hline
Checkerboard              & line-graph      & 435         & 18521     & 393              &  67              & 67           \\
                          & SABRE           & 2121        & 23060     & 282              &  1750            & 109 $\pm$ 1  \\
  \end{tabular}
  \caption{\label{tab:benchmark} Excerpt of the benchmarking data available in the Supplemental Material~\cite{Note2}. The column headers are defined in the main text.
  }
\end{table*}

In the Supplemental Material \footnote{Supplemental Material available at [URL]. In the Supplemental Material, we implement, showcase and benchmark line-graph routing for various graphs. The Supplemental Material is also available at Ref.~\cite{github}}, we implement, showcase and benchmark line-graph routing, Algorithm~\ref{alg:line-graph_routing}, together with the removal of lone leafs (\app{lone_leafs}). The implementation takes any Qiskit~\cite{qiskit} quantum circuit consisting of one- and two-qubit gates, constructs its coupling graph $L(G)$ or reports that the coupling graph is not a line graph, finds $G$ and $\heavy(G)$, and outputs the routed circuit with coupling graph $\heavy(G)$. The implementation does not rely heavily on Qiskit's methods and can hence straightforwardly be transformed to an implementation in other quantum software development kits.

Line-graph routing is benchmarked against all routing methods available in Qiskit by default~\cite{qiskit,Note2}. In this section, we show line-graph routing is able to confidently outperform these default methods on relevant problem instances. There are also problem instances where line-graph routing does not outperform the default methods. In the end of this section, we discuss for which types of instances we expect line-graph routing to be superior. We consider two types of circuits.

\emph{(i) Random.} With probability 2/5, a CNOT gate is placed along a randomly chosen edge of a given virtual graph. With a probability 3/5 a gate from the set $\{H,S,T\}$ is chosen uniformly at random and placed at a random node.

\emph{(ii) Quantum simulation.} As described in detail for the kagome lattice (\sec{quant_sim}), circuits for the dynamical and static quantum simulation of the HAFM on any lattice can be defined by an edge-coloring of that lattice~\cite{kattemolle2022variational,burkard2022recipes}. An edge coloring of a graph is an assignment of colors to the edges such that no two edges with the same color are incident on the same node. This edge coloring is called minimal if it uses the least possible number of colors. We perform edge coloring of the virtual lattices by an automatic method that generally finds an edge coloring that is not minimal. The benefit of this coloring method is that it does not require a manual assignment of edge colors. The downside is that we expect line-graph routing to work best (compared to other methods) for circuits derived from a minimal edge coloring. This does not pose a problem because, as we show in this section, line-graph routing is already able to outperform the default methods in Qiskit in routing circuits derived from a nonminimal edge coloring.

We found SABRE~\cite{li2019mapping} to outperform the other methods in Qiskit and therefore we focus on a comparison between line-graph routing and SABRE in what follows. Unlike line-graph routing, SABRE is a probabilistic routing method that obtains a different qubit routing with each run. Additionally, the intensity of the optimization that is part of SABRE can be varied, leading to a tradeoff between the classical resources required and the performance characteristics of the routed circuits. We address these issues by running SABRE 16 times (at fixed optimization level) and comparing the performance against one run of line-graph routing. We do this separately for every optimization level available by default in Qiskit, which range from optimization level 0 (`no optimization') to optimization level 3 (`heavy-weight optimization') in integer steps~\cite{qiskit}.

In Table \ref{tab:benchmark}, we show an excerpt of the benchmarking data, which includes problem instances on which line-graph routing does and does not perform well. The following performance characteristics are listed.

\emph{(i) Opt. depth.} The optimal (lowest) depth reached by the routing method among all runs (line-graph routing is run once per virtual graph, SABRE is run 16 times per virtual graph). Routed circuits are obtained by inserting $\SWAP$ gates (as dictated by line-graph routing or SABRE) and no further gate identities are used to simplify the resulting circuits. So, in the case of random input circuits, the routed circuits consist of gates from the set $\{\CNOT, H,S ,T\} \cup \{\SWAP\}$. The routed circuits contain, for example, double $H$ gates if those were present in the input circuit. In the case of quantum simulation input circuits, the routed circuits consist of gates from the set $\{\mrm{SINGLET}, \HEIS(\al)\} \cup \{\SWAP\}$.

\emph{(ii) $n_\SWAP$.} The number of $\SWAP$ gates of the routed circuit that achieved the lowest depth.

\emph{(iii) $n_\mrm{qubit}$.} The number of active qubits in the routed circuit that achieved the lowest depth.

\emph{(iv) $t_\mrm{tot}$.} The total wall-clock time needed to run all runs of the routing method. For line-graph routing, this includes the time needed to find $\heavy(G)$ from $L(G)$. The implementation of line-graph routing repeatedly loops through all gates using (slow) Python loops and can likely be sped up considerably, if needed. We use Qiskit's standard implementation of SABRE. SABRE is given the target graph $\heavy(G)$ as input and hence finding $\heavy(G)$ from $L(G)$ is not included in its wall-clock time. The benchmarks for different methods are always run on the same hardware.

\emph{(v) $\bar t$.} The average wall-clock time for a single run of the routing method. The error bars are calculated by bootstrapping the individual wall-clock times and represent symmetrized 95\% confidence intervals.

The first two data lines of Table \ref{tab:benchmark} show the performance characteristics of line-graph routing and SABRE when applied to a random circuit on the complete graph with 9 nodes. SABRE was run with optimization level 1. (Passing a higher optimization level to Qiskit's transpiler will trigger the usage of gate identities.) Already at an optimization level of 1 SABRE outperforms line-graph routing on all performance characteristics considered except the total wall-clock time.

Data lines 3 and 4 of Table \ref{tab:benchmark} show the performance characteristics of line-graph routing and SABRE when applied to circuits for the quantum simulation of the HAFM on patches of the kagome lattice measuring $7\times 7$ unit cells, with open boundary conditions and padded edges (see~\cite{Note2}). Here, SABRE was run at Qiskit's transpiler optimization level 3, the highest optimization level available. As opposed to the transpilation of the random circuits, no gate identities are used in this process because the gates $\{\SWAP, \mrm{SINGLET}, \HEIS\}$ have unknown properties to Qiskit's transpiler~\cite{qiskit}. Line-graph routing outperforms SABRE in terms of the optimal depth of the routed circuits by about a factor of 3.5, while at the same time requiring less time than one run of SABRE. The line-graph routed circuit also uses less $\SWAP$ gates than the SABRE routed circuit with the lowest depth. However, the line-graph routed circuit uses a factor of 1.5 more qubits than the SABRE routed circuit that achieved the lowest depth. Nevertheless, the space-time volume of the line-graph routed circuit is still about half of the space-time volume of the SABRE routed circuit with the lowest depth. Similar results, with an even larger performance gap, hold for the routing of circuits for the quantum simulation of the HAFM on patches of the shuriken lattice measuring $7 \times 7$ unit cells (lines 5 and 6 of Table \ref{tab:benchmark}) and patches of the checkerboard lattice of $7.5\times 7.5$ unit cells (last two lines of \tab{benchmark}).

Line-graph routing was conceived while keeping in mind its application in mapping quantum circuits to hardware with a lattice-like low-degree coupling graph. It is therefore not expected to perform well in mapping quantum circuits to hardware with a high-degree coupling graph without lattice-like structure. In fact, the benchmarking results show that line-graph routing is not well-suited for mapping circuits on the complete coupling graph to hardware with a star coupling graph. One property of line-graph routing that leads to its low effectivity on unstructured, high-degree graphs is that in the output circuit of line-graph routing, the qubits are assigned a base location where they return to after they are acted on by one [\eq{mediated_gate}] or multiple [\eq{cancel_SWAP}] gates from the input circuit. This is likely advantageous for input quantum circuits with a structured, low-degree coupling graph. For example, in \fig{circuit}, in the third layer of the cycle (red gates), all qubits are still close to where they are needed, leading to the insertion of few SWAP gates to get them there.

However, on high-degree coupling graphs, the property of a base location need not be advantageous since any qubit can be routed to any other qubit in relatively few steps and regularly returning qubits to their base location leads to the insertion of unnecessary SWAP gates. As an extreme example, let us look at the action of line-graph routing on a circuit $C$ with a complete coupling graph, where at one point in $C$ a cascade of CNOT gates, $\prod_{i=n}^1\CNOT_{i,i+1}$, is prescribed. Line-graph routing maps the circuit $C$ to a circuit on the star graph (\tab{examples}). To perform $\CNOT_{i,i+1}$ on star-graph hardware, line-graph routing first swaps qubit $i$ to the center of the star, performs a CNOT between the center qubit and qubit $i+1$, and swaps qubit $i$ back to its original position [\eq{mediated_gate}]. Not insisting that the qubits eventually return to their original position leads to the possibility of more efficient routing strategy. To perform the $\CNOT_{i,i+1}$, swap qubit $i$ to the center, perform the CNOT between the center qubit and qubit $i+1$, and leave qubit $i$ in the center. After repeating this procedure for the subsequent CNOT in the circuit, qubit $i$ will end up at the initial location of qubit $i+1$. The latter approach uses half of the number of SWAP gates compared to line-graph routing.

A second situation in which line-graph routing is not expected to perform well is when there are subgraphs of the virtual graph $L(G)$ that are isomorphic to subgraphs of $\heavy(G)$. In that case, unnecessary mediating qubits may be inserted. For example, let $L(G)$ be the path graph on $n$ nodes. When applied to circuits on this graph, line-graph routing introduces mediating qubits between all qubits of $L(G)$, leading to a circuit on a path graph with $n-1$ added mediating qubits. This happens despite the fact that no routing is needed at all. If needed, line-graph routing can be enhanced such that it detects and prevents this behavior.

\section{Discussion and outlook}\label{sec:discussion}
In this paper, we developed line-graph routing, a qubit routing strategy that efficiently and deterministically maps any quantum circuit on a line graph $L(G)$ to a circuit on the heavy graph $\heavy(G)$. By software implementation and benchmarking, we showed its ability to outperform standard, general-purpose methods on input quantum circuits, circuit sizes, and hardware connectivity graphs of practical relevance.

Line-graph routing showed not to perform well in mapping circuits on the complete graph to the star graph. We attribute this to the high degree and absence of a lattice-like structure of the complete graph. Based on our benchmarking results, we expect line-graph routing to outperform general-purpose methods in routing circuits with a line-graph coupling graph to low-degree hardware coupling graphs. For superconducting qubits, these are exactly the hardware coupling graphs preferable from an engineering standpoint~\cite{chamberland2020heavylattices,nation_paik_cross_nazario_2022, bravyi2022quantumcentric}.

Line-graph routing is limited in its applicability because it is only defined for pairs of graphs $(L(G),\heavy(G))$. General-purpose methods are able to map any circuit on any graph to a circuit on any other graph (given that the latter graph is connected and has at least the same number of nodes as the former graph). However, in general, finding the optimal routing strategy is NP-hard. It is therefore unlikely that general-purpose methods can find the optimal or close-to-optimal routing strategy. Routing strategies that are defined on a subset of possible input circuits can nevertheless be highly efficient and effective, and indeed line-graph routing is shown to have the ability to outperform standard-purpose methods on the subset of problem instances for which it is defined. The set of pairs of graphs $(L(G),\heavy(G))$ is still infinitely large and contains pairs of practical relevance.

The routed circuit for the simulation of a nearest-neighbor model on $L(G)$ can be reinterpreted as  a circuit for the quantum simulation of a next-nearest-neighbor model on $\heavy(G)$. This is evident from \tab{examples} and \fig{four_coloring} and follows directly from the definition of line graphs as given in \sec{general_case}. This routed circuit for the simulation of a nearest-neighbor model on $L(G)$ can be lifted to a circuit for the quantum simulation of a model containing both nearest- and next-nearest-neighbor interactions on $\heavy(G)$ by adding the gates arising from nearest-neighbor interactions on $\heavy(G)$ to the routed circuit. These gates naturally satisfy the $\heavy(G)$ connectivity graph and can thus be added to the circuit without any routing. This further enlarges the range of applicability of line-graph routing.

The effectiveness of line-graph routing may be improved further by leveraging the freedom of which qubit is swapped with the mediating qubit. In line-graph routing, for the implementation of the gate $U^\ell_{ij}$, it is qubit $i$ that is swapped with the mediating qubit [\eq{mediated_gate}]. However, in some cases, it is only when swapping qubit $j$ with the mediating qubit that a cancellation of SWAP gates from consecutive mediated two-qubit gates $MU$ [\eq{mediated_gate}, \fig{circuit}] occurs. Thus, line-graph routing may be improved by letting the decision of which qubit to swap with the mediating qubit depend on the ensuing gates in the input quantum circuit.

Although we showed the ability of line-graph routing to outperform general-purpose methods, we did not prove it gives the optimal routing. In fact, it was shown to be suboptimal in cases where the input circuit is a circuit on the complete graph. A proof or refutation for the optimality of line-graph routing (possibly after including the improvement of the previous paragraph) would therefore require careful consideration of the allowed input circuits and performance characteristic for which optimality is considered. Such a proof may inspire even stronger methods that build on or generalize line-graph routing.

A straightforward but exciting way to carry on the work in the current paper is to run our routed circuits on actual quantum hardware. The circuits available in the Supplemental Material~\cite{Note2} can be executed as-is on hardware with the appropriate hardware coupling graphs. Remaining challenges therein are the extraction of useful physical quantities from the generated quantum states. To obtain results that go beyond anything that can be obtained classically requires further improvements of error mitigation techniques and quantum hardware.
\ \\ 

All code and data used to generate the results in this paper are available as Supplemental Material~\cite{Note2} and at Ref.~\cite{github}.

\begin{acknowledgements}
JK acknowledges funding from the Competence Cen-
ter Quantum Computing Baden-W\"urttemberg, under the
project QORA. Benchmarking was carried out on
the Scientific Compute Cluster of the University of Konstanz (SCCKN).
\end{acknowledgements}

\bibliography{bib.bib}

\appendix

\section{Removal of lone leaf nodes}\label{sec:lone_leafs}
In some cases, line-graph routing produces circuits containing qubits $i$ that are swapped with only one corresponding mediating qubit $m_i$ during the course of the entire circuit. In these cases, the mediating qubits $m_i$ can be removed and replaced by the corresponding qubits $i$, leading to a reduction in the qubit and SWAP count of the routed circuit. 

We define a \emph{lone leaf} $i$ as a node in a coupling graph $\heavy(G)$ that is of degree one and shares its only edge with a node that has no other neighbors of degree one. We do not assume the coupling graph is a tree. As an example, consider the triangle graph with an extra node $i$ connected to one of the nodes $m_i$ of the triangle. The node $i$ is a lone leaf.

If a node $i$ is a lone leaf in a coupling graph $\heavy(G)$, by the construction of line-graph routing its neighbor $m$ must be a mediating node. For every two-qubit gate that is performed between qubit $i$ and any other neighbor $j$ of $m$, qubits $i$ and $m$ first need to be swapped. So, as far as the interactions with qubit $i$ are concerned, qubit $m$ may be fully eliminated by simply removing qubit $m$ and reverting mediated two-qubit gates $MU_{imj}$ to the original gates $MU_{imj}$. Physically, qubit $i$ can be put on the place of qubit $m$. After the removal of $m$ and the relocation of qubit $i$, inspection of \eq{mediated_gate} reveals that qubit $i$ may in fact act as a mediating node for any two-qubit gates between any two neighbors $j,k$ of $m$ unequal to $i$; after a mediated gate $MU_{jik}$, qubit $i$ returns to its starting position unaffected. 
\end{document}


\onecolumn
    \maketitle
    \begin{abstract}

We showcase and benchmark line-graph qubit routing by routing random
circuits (with fixed virtual graph) and circuits needed for the quantum
simulation of the Heisenberg antiferromagnet on various graphs. The code
implementing line-routing itself is found in
\texttt{line\_graph\_routing.py} and maps Qiskit quantum circuits to
Qiskit quantum circuits. (This can be altered to other circuit libraries
with little effort.) We benchmark our results against other methods. The
interactive version of this document is
\texttt{line\_graph\_routing.ipynb}.

\end{abstract}

\tableofcontents

\clearpage

\twocolumn
    \hypertarget{requirements}{%
\section{Requirements}\label{requirements}}

    This notebook should typically run after installing the following
packages with pip (or conda). In a terminal, run

\begin{verbatim}
pip install qiskit[visualization]
\end{verbatim}

or

\begin{verbatim}
pip install 'qiskit[visualization]'
\end{verbatim}

and

\begin{verbatim}
pip install netket networkx tabulate
\end{verbatim}

Note Netket currently needs Python 3.9 (and SciPy >=
1.9.3). Netket is only used to generate patches of the kagome lattice as
graphs and not for line-graph routing itself. This notebook was tested
with a pip environment that can be recreated with
\texttt{requirements.txt} by running
\texttt{pip\ install\ -r\ requirements.txt} (after creating a new
environment).

The file \texttt{line\_graph\_routing.py} should be placed in the same
folder as the current notebook.

    \begin{tcolorbox}[breakable, size=fbox, boxrule=1pt, pad at break*=1mm,colback=cellbackground, colframe=cellborder]
\prompt{In}{incolor}{1}{\boxspacing}
\begin{Verbatim}[commandchars=\\\{\}]
\PY{k+kn}{import} \PY{n+nn}{line\PYZus{}graph\PYZus{}routing} \PY{k}{as} \PY{n+nn}{lgr}
\PY{k+kn}{import} \PY{n+nn}{networkx} \PY{k}{as} \PY{n+nn}{nx}
\end{Verbatim}
\end{tcolorbox}

    \hypertarget{kagome-to-heavy-hex}{%
\section{Kagome to heavy-hex}\label{kagome-to-heavy-hex}}

    \hypertarget{random}{%
\subsection{Random}\label{random}}

    Create a random circuit on a patch of the kagome lattice of
\(3\times 3\) unit cells and show the circuit's coupling graph. With
probability 2/5, a CNOT is placed along an edge of the connectivity
graph. With a probability 3/5 a gate from \{H,S,T\} is chosen uniformly
at random and placed at a random node.

    \begin{tcolorbox}[breakable, size=fbox, boxrule=1pt, pad at break*=1mm,colback=cellbackground, colframe=cellborder]
\prompt{In}{incolor}{6}{\boxspacing}
\begin{Verbatim}[commandchars=\\\{\}]
\PY{n}{lg} \PY{o}{=} \PY{n}{lgr}\PY{o}{.}\PY{n}{kagome}\PY{p}{(}\PY{l+m+mi}{3}\PY{p}{,} \PY{l+m+mi}{3}\PY{p}{)}
\PY{n}{qc} \PY{o}{=} \PY{n}{lgr}\PY{o}{.}\PY{n}{random\PYZus{}circuit}\PY{p}{(}\PY{n}{lg}\PY{p}{,} \PY{l+m+mi}{10}\PY{o}{*}\PY{o}{*}\PY{l+m+mi}{4}\PY{p}{)}
\PY{n}{cg} \PY{o}{=} \PY{n}{lgr}\PY{o}{.}\PY{n}{coupling\PYZus{}graph}\PY{p}{(}\PY{n}{qc}\PY{p}{)}
\PY{n}{nx}\PY{o}{.}\PY{n}{draw\PYZus{}kamada\PYZus{}kawai}\PY{p}{(}\PY{n}{cg}\PY{p}{)}
\PY{n+nb}{print}\PY{p}{(}\PY{n}{qc}\PY{o}{.}\PY{n}{depth}\PY{p}{(}\PY{p}{)}\PY{p}{)}
\end{Verbatim}
\end{tcolorbox}

    \begin{Verbatim}[commandchars=\\\{\}]
755
    \end{Verbatim}

    \begin{center}
    \adjustimage{max size={0.9\linewidth}{0.9\paperheight}}{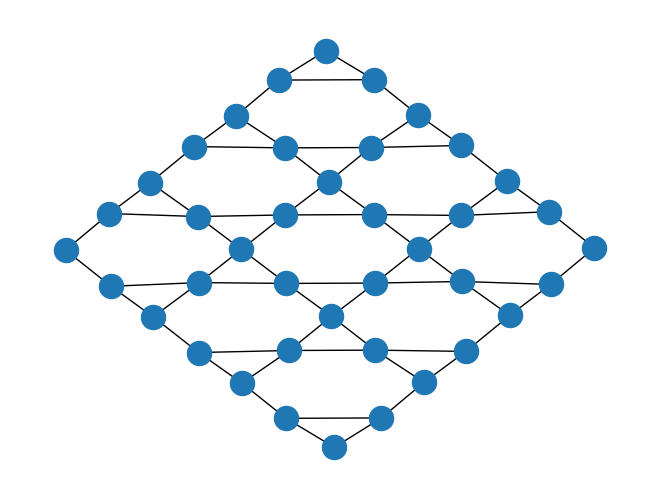}
    \end{center}

    Route the circuit to a circuit with heavy-hex coupling graph.

    \begin{tcolorbox}[breakable, size=fbox, boxrule=1pt, pad at break*=1mm,colback=cellbackground, colframe=cellborder]
\prompt{In}{incolor}{7}{\boxspacing}
\begin{Verbatim}[commandchars=\\\{\}]
\PY{n}{qc} \PY{o}{=} \PY{n}{lgr}\PY{o}{.}\PY{n}{line\PYZus{}graph\PYZus{}route}\PY{p}{(}\PY{n}{qc}\PY{p}{)}
\PY{n}{cg} \PY{o}{=} \PY{n}{lgr}\PY{o}{.}\PY{n}{coupling\PYZus{}graph}\PY{p}{(}\PY{n}{qc}\PY{p}{)}
\PY{n}{nx}\PY{o}{.}\PY{n}{draw\PYZus{}kamada\PYZus{}kawai}\PY{p}{(}\PY{n}{cg}\PY{p}{)}
\PY{n+nb}{print}\PY{p}{(}\PY{n}{qc}\PY{o}{.}\PY{n}{depth}\PY{p}{(}\PY{p}{)}\PY{p}{)}
\end{Verbatim}
\end{tcolorbox}

    \begin{Verbatim}[commandchars=\\\{\}]
1337
    \end{Verbatim}

    \begin{center}
    \adjustimage{max size={0.9\linewidth}{0.9\paperheight}}{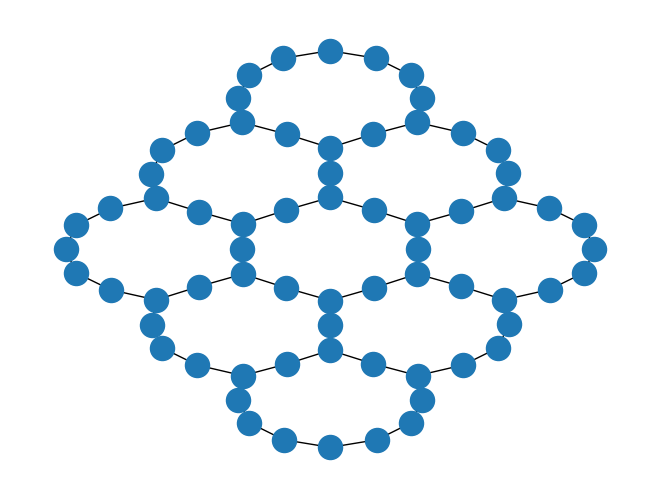}
    \end{center}

    \hypertarget{quantum-simulation}{%
\subsection{Quantum simulation}\label{quantum-simulation}}

We base our circuits on edge colorings of the (kagome) lattice by
identifying every color with a layer of \(\mathrm{HEIS}\) gates. One of
the colors (color `0') doubles as a color specifying the initial state
by indicating along which edges singlet states are placed. The entire
circuit is repeated \(p\) times, excluding initial state preparation.
Every \(\mathrm{HEIS}\) gets its own parameter. These parameters can
later be bound to specific values to obtain circuits for dynamical
quantum simulation or for simulated adiabatic state preparation.

First, create and show an edge coloring of the kagome lattice.

    \begin{tcolorbox}[breakable, size=fbox, boxrule=1pt, pad at break*=1mm,colback=cellbackground, colframe=cellborder]
\prompt{In}{incolor}{8}{\boxspacing}
\begin{Verbatim}[commandchars=\\\{\}]
\PY{n}{lg} \PY{o}{=} \PY{n}{lgr}\PY{o}{.}\PY{n}{edge\PYZus{}coloring}\PY{p}{(}\PY{n}{lg}\PY{p}{)}
\PY{n}{lgr}\PY{o}{.}\PY{n}{draw\PYZus{}edge\PYZus{}coloring}\PY{p}{(}\PY{n}{lg}\PY{p}{)}
\end{Verbatim}
\end{tcolorbox}

    \begin{Verbatim}[commandchars=\\\{\}]
Matching is perfect
Edge coloring is not minimal
    \end{Verbatim}

    \begin{center}
    \adjustimage{max size={0.9\linewidth}{0.9\paperheight}}{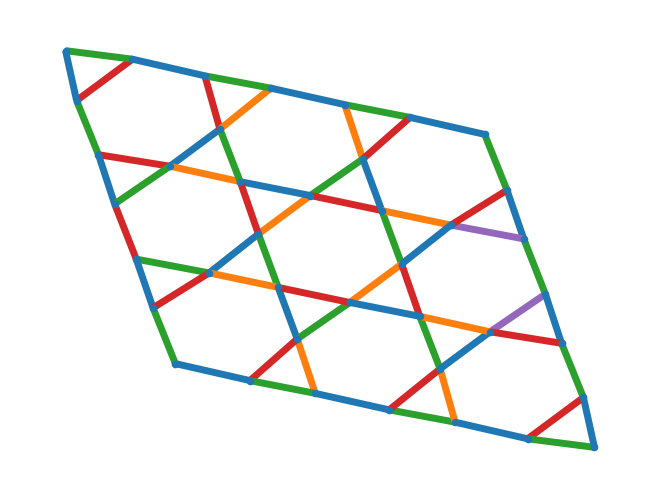}
    \end{center}

    Create the associated circuit, route it to heavy-hex hardware and show
the coupling graph of the routed circuit.

    \begin{tcolorbox}[breakable, size=fbox, boxrule=1pt, pad at break*=1mm,colback=cellbackground, colframe=cellborder]
\prompt{In}{incolor}{9}{\boxspacing}
\begin{Verbatim}[commandchars=\\\{\}]
\PY{n}{p} \PY{o}{=} \PY{l+m+mi}{1}
\PY{n}{qc} \PY{o}{=} \PY{n}{lgr}\PY{o}{.}\PY{n}{heis\PYZus{}circuit}\PY{p}{(}\PY{n}{lg}\PY{p}{,} \PY{n}{p}\PY{p}{)}
\PY{n+nb}{print}\PY{p}{(}\PY{n}{qc}\PY{o}{.}\PY{n}{depth}\PY{p}{(}\PY{p}{)}\PY{p}{)}
\PY{n}{qc} \PY{o}{=} \PY{n}{lgr}\PY{o}{.}\PY{n}{line\PYZus{}graph\PYZus{}route}\PY{p}{(}\PY{n}{qc}\PY{p}{)}
\PY{n+nb}{print}\PY{p}{(}\PY{n}{qc}\PY{o}{.}\PY{n}{depth}\PY{p}{(}\PY{p}{)}\PY{p}{)}
\PY{n}{cg} \PY{o}{=} \PY{n}{lgr}\PY{o}{.}\PY{n}{coupling\PYZus{}graph}\PY{p}{(}\PY{n}{qc}\PY{p}{)}
\PY{n}{nx}\PY{o}{.}\PY{n}{draw\PYZus{}kamada\PYZus{}kawai}\PY{p}{(}\PY{n}{cg}\PY{p}{)}
\end{Verbatim}
\end{tcolorbox}

    \begin{Verbatim}[commandchars=\\\{\}]
6
15
    \end{Verbatim}

    \begin{center}
    \adjustimage{max size={0.9\linewidth}{0.9\paperheight}}{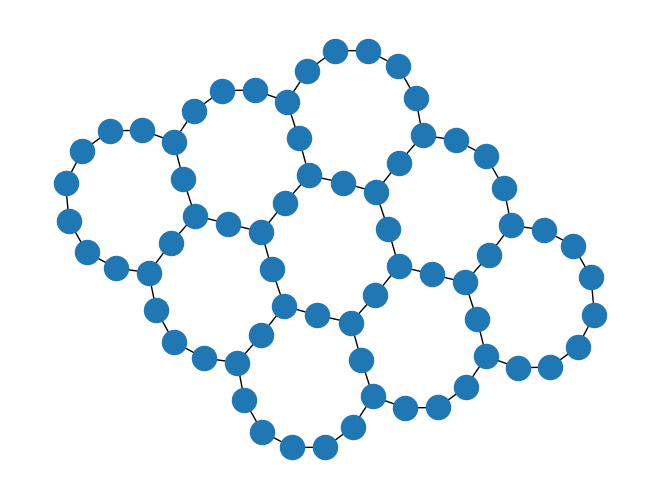}
    \end{center}

    Since the above circuit has the connectivity of a 2D graph, the
corresponding quantum circuit diagram will not be very insightful.
However, for a single unit cell of the kagome lattice, the routed
quantum circuit becomes a circuit on a circle, which allows for a clear
representation as a quantum circuit diagram. We first create an edge
coloring of the unit cell patch.

    \begin{tcolorbox}[breakable, size=fbox, boxrule=1pt, pad at break*=1mm,colback=cellbackground, colframe=cellborder]
\prompt{In}{incolor}{10}{\boxspacing}
\begin{Verbatim}[commandchars=\\\{\}]
\PY{n}{lg} \PY{o}{=} \PY{n}{lgr}\PY{o}{.}\PY{n}{kagome}\PY{p}{(}\PY{l+m+mi}{1}\PY{p}{,} \PY{l+m+mi}{1}\PY{p}{)}
\PY{n}{lg} \PY{o}{=} \PY{n}{lgr}\PY{o}{.}\PY{n}{edge\PYZus{}coloring}\PY{p}{(}\PY{n}{lg}\PY{p}{)}
\PY{n}{lgr}\PY{o}{.}\PY{n}{draw\PYZus{}edge\PYZus{}coloring}\PY{p}{(}\PY{n}{lg}\PY{p}{)}
\end{Verbatim}
\end{tcolorbox}

    \begin{Verbatim}[commandchars=\\\{\}]
Matching is perfect
Edge coloring is not minimal
    \end{Verbatim}

    \begin{center}
    \adjustimage{max size={0.9\linewidth}{0.9\paperheight}}{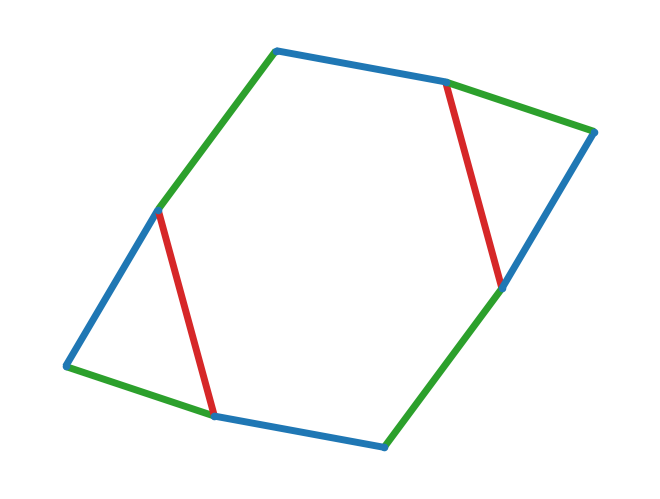}
    \end{center}

    Create the HEIS circuit ased on this coloring, map it to heavy-hex
hardware, and show the coupling graph of the routed circuit.

    \begin{tcolorbox}[breakable, size=fbox, boxrule=1pt, pad at break*=1mm,colback=cellbackground, colframe=cellborder]
\prompt{In}{incolor}{11}{\boxspacing}
\begin{Verbatim}[commandchars=\\\{\}]
\PY{n}{p} \PY{o}{=} \PY{l+m+mi}{1}
\PY{n}{qc} \PY{o}{=} \PY{n}{lgr}\PY{o}{.}\PY{n}{heis\PYZus{}circuit}\PY{p}{(}\PY{n}{lg}\PY{p}{,} \PY{n}{p}\PY{p}{)}
\PY{n+nb}{print}\PY{p}{(}\PY{n}{qc}\PY{o}{.}\PY{n}{depth}\PY{p}{(}\PY{p}{)}\PY{p}{)}
\PY{n}{qc} \PY{o}{=} \PY{n}{lgr}\PY{o}{.}\PY{n}{line\PYZus{}graph\PYZus{}route}\PY{p}{(}\PY{n}{qc}\PY{p}{)}
\PY{n+nb}{print}\PY{p}{(}\PY{n}{qc}\PY{o}{.}\PY{n}{depth}\PY{p}{(}\PY{p}{)}\PY{p}{)}
\PY{n}{cg} \PY{o}{=} \PY{n}{lgr}\PY{o}{.}\PY{n}{coupling\PYZus{}graph}\PY{p}{(}\PY{n}{qc}\PY{p}{)}
\PY{n}{nx}\PY{o}{.}\PY{n}{draw\PYZus{}kamada\PYZus{}kawai}\PY{p}{(}\PY{n}{cg}\PY{p}{,} \PY{n}{with\PYZus{}labels} \PY{o}{=} \PY{l+s+s1}{\PYZsq{}}\PY{l+s+s1}{true}\PY{l+s+s1}{\PYZsq{}}\PY{p}{)}
\end{Verbatim}
\end{tcolorbox}

    \begin{Verbatim}[commandchars=\\\{\}]
4
7
    \end{Verbatim}

    \begin{center}
    \adjustimage{max size={0.9\linewidth}{0.9\paperheight}}{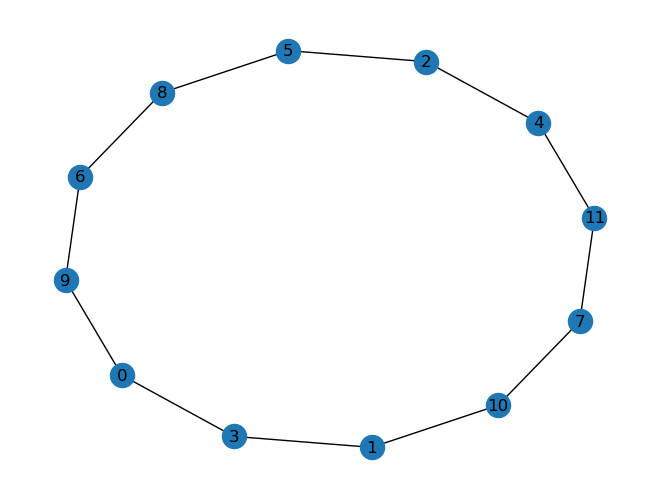}
    \end{center}

    Show the circuit diagram of the routed circuit, with parameters
\texttt{al\_i}.

    \begin{tcolorbox}[breakable, size=fbox, boxrule=1pt, pad at break*=1mm,colback=cellbackground, colframe=cellborder]
\prompt{In}{incolor}{12}{\boxspacing}
\begin{Verbatim}[commandchars=\\\{\}]
\PY{n}{wo} \PY{o}{=} \PY{p}{[}\PY{l+m+mi}{0}\PY{p}{,} \PY{l+m+mi}{3}\PY{p}{,} \PY{l+m+mi}{1}\PY{p}{,} \PY{l+m+mi}{10}\PY{p}{,} \PY{l+m+mi}{7}\PY{p}{,} \PY{l+m+mi}{11}\PY{p}{,} \PY{l+m+mi}{4}\PY{p}{,} \PY{l+m+mi}{2}\PY{p}{,} \PY{l+m+mi}{5}\PY{p}{,} \PY{l+m+mi}{8}\PY{p}{,} \PY{l+m+mi}{6}\PY{p}{,} \PY{l+m+mi}{9}\PY{p}{]} \PY{c+c1}{\PYZsh{} In the circuit diagram, place qubits in this order.}
\PY{n}{qc}\PY{o}{.}\PY{n}{draw}\PY{p}{(}\PY{l+s+s1}{\PYZsq{}}\PY{l+s+s1}{latex}\PY{l+s+s1}{\PYZsq{}}\PY{p}{,} \PY{n}{wire\PYZus{}order} \PY{o}{=} \PY{n}{wo}\PY{p}{)}
\end{Verbatim}
\end{tcolorbox}

\prompt{Out}{outcolor}{12}{}
    
    \begin{center}
    \adjustimage{max size={0.9\linewidth}{0.9\paperheight}}{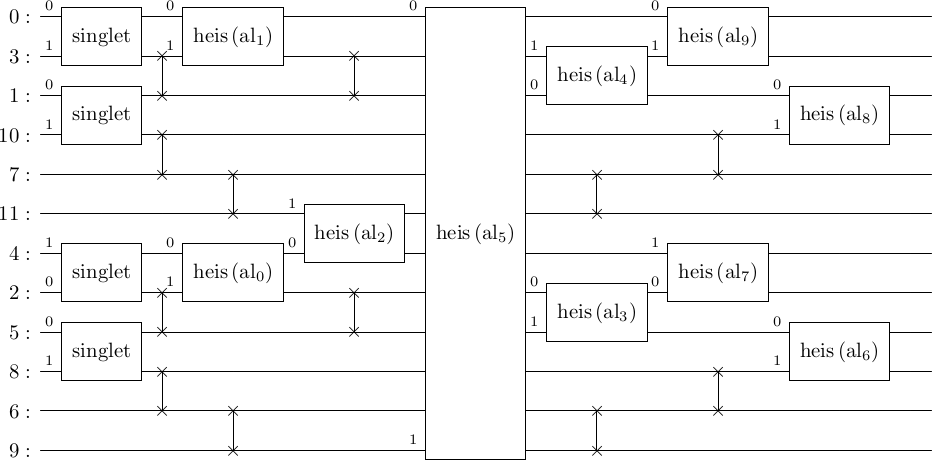}
    \end{center}

    The fircuit depth can be reduced further by replacement of the initial
and final SWAP gates between qubits (10,7) and (8,6) by a relabeling of
those qubits.

    \hypertarget{complete-graph-to-star-graph}{%
\section{Complete graph to star
graph}\label{complete-graph-to-star-graph}}

    \hypertarget{random}{%
\subsection{Random}\label{random}}

    Create a random circuit on the complete graph of four nodes and show the
circuit's coupling graph.

    \begin{tcolorbox}[breakable, size=fbox, boxrule=1pt, pad at break*=1mm,colback=cellbackground, colframe=cellborder]
\prompt{In}{incolor}{2}{\boxspacing}
\begin{Verbatim}[commandchars=\\\{\}]
\PY{n}{n} \PY{o}{=} \PY{l+m+mi}{5}
\PY{n}{lg} \PY{o}{=} \PY{n}{nx}\PY{o}{.}\PY{n}{complete\PYZus{}graph}\PY{p}{(}\PY{l+m+mi}{4}\PY{p}{)}
\PY{n}{qc} \PY{o}{=} \PY{n}{lgr}\PY{o}{.}\PY{n}{random\PYZus{}circuit}\PY{p}{(}\PY{n}{lg}\PY{p}{,} \PY{l+m+mi}{10}\PY{o}{*}\PY{o}{*}\PY{l+m+mi}{2}\PY{p}{)}
\PY{n}{cg} \PY{o}{=} \PY{n}{lgr}\PY{o}{.}\PY{n}{coupling\PYZus{}graph}\PY{p}{(}\PY{n}{qc}\PY{p}{)}
\PY{n}{nx}\PY{o}{.}\PY{n}{draw\PYZus{}kamada\PYZus{}kawai}\PY{p}{(}\PY{n}{cg}\PY{p}{)}
\PY{n+nb}{print}\PY{p}{(}\PY{n}{qc}\PY{o}{.}\PY{n}{depth}\PY{p}{(}\PY{p}{)}\PY{p}{)}
\end{Verbatim}
\end{tcolorbox}

    \begin{Verbatim}[commandchars=\\\{\}]
56
    \end{Verbatim}

    \begin{center}
    \adjustimage{max size={0.9\linewidth}{0.9\paperheight}}{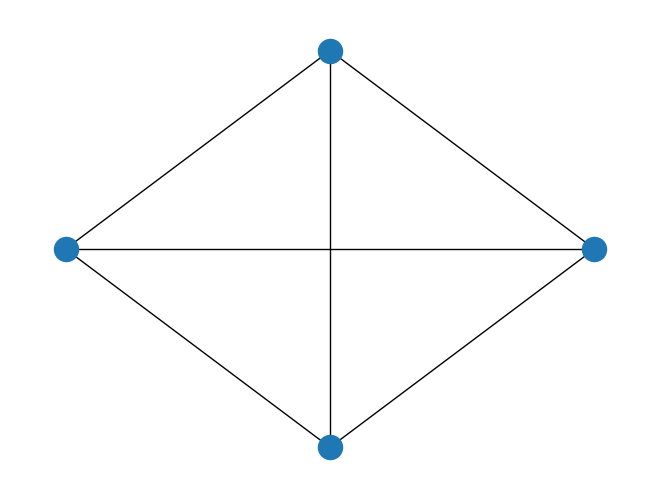}
    \end{center}

    Route the circuit to a circuit with star-graph connectivity.

    \begin{tcolorbox}[breakable, size=fbox, boxrule=1pt, pad at break*=1mm,colback=cellbackground, colframe=cellborder]
\prompt{In}{incolor}{3}{\boxspacing}
\begin{Verbatim}[commandchars=\\\{\}]
\PY{n}{qc} \PY{o}{=} \PY{n}{lgr}\PY{o}{.}\PY{n}{line\PYZus{}graph\PYZus{}route}\PY{p}{(}\PY{n}{qc}\PY{p}{)}
\PY{n}{cg} \PY{o}{=} \PY{n}{lgr}\PY{o}{.}\PY{n}{coupling\PYZus{}graph}\PY{p}{(}\PY{n}{qc}\PY{p}{)}
\PY{n}{nx}\PY{o}{.}\PY{n}{draw\PYZus{}kamada\PYZus{}kawai}\PY{p}{(}\PY{n}{cg}\PY{p}{)}
\PY{n+nb}{print}\PY{p}{(}\PY{n}{qc}\PY{o}{.}\PY{n}{depth}\PY{p}{(}\PY{p}{)}\PY{p}{)}
\end{Verbatim}
\end{tcolorbox}

    \begin{Verbatim}[commandchars=\\\{\}]
108
    \end{Verbatim}

    \begin{center}
    \adjustimage{max size={0.9\linewidth}{0.9\paperheight}}{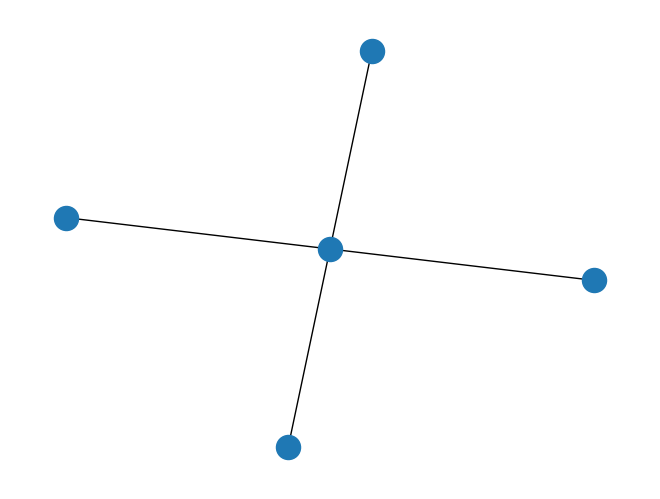}
    \end{center}

    \hypertarget{quantum-simulation}{%
\subsection{Quantum simulation}\label{quantum-simulation}}

As before, circuits are defined by identifying every color with a layer
of \(\mathrm{HEIS}\)-gates. For more details, see the
\texttt{kagome\ to\ heavy-hex} section.

Create and show an edge coloring of the complete graph.

    \begin{tcolorbox}[breakable, size=fbox, boxrule=1pt, pad at break*=1mm,colback=cellbackground, colframe=cellborder]
\prompt{In}{incolor}{4}{\boxspacing}
\begin{Verbatim}[commandchars=\\\{\}]
\PY{n}{lg} \PY{o}{=} \PY{n}{lgr}\PY{o}{.}\PY{n}{edge\PYZus{}coloring}\PY{p}{(}\PY{n}{lg}\PY{p}{)}
\PY{n}{lgr}\PY{o}{.}\PY{n}{draw\PYZus{}edge\PYZus{}coloring}\PY{p}{(}\PY{n}{lg}\PY{p}{)}
\end{Verbatim}
\end{tcolorbox}

    \begin{Verbatim}[commandchars=\\\{\}]
Matching is perfect
Edge coloring is not minimal
    \end{Verbatim}

    \begin{center}
    \adjustimage{max size={0.9\linewidth}{0.9\paperheight}}{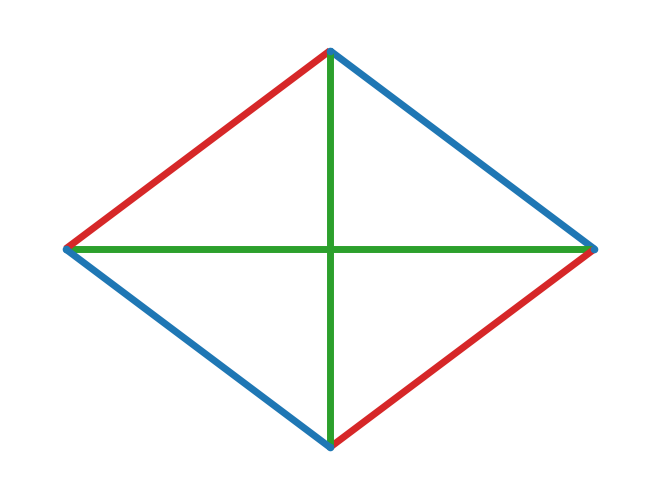}
    \end{center}

    Create the associated circuit, route it to heavy-hex hardware, and show
the coupling graph of the routed circuit.

    \begin{tcolorbox}[breakable, size=fbox, boxrule=1pt, pad at break*=1mm,colback=cellbackground, colframe=cellborder]
\prompt{In}{incolor}{5}{\boxspacing}
\begin{Verbatim}[commandchars=\\\{\}]
\PY{n}{p} \PY{o}{=} \PY{l+m+mi}{1}
\PY{n}{qc} \PY{o}{=} \PY{n}{lgr}\PY{o}{.}\PY{n}{heis\PYZus{}circuit}\PY{p}{(}\PY{n}{lg}\PY{p}{,} \PY{n}{p}\PY{p}{)}
\PY{n+nb}{print}\PY{p}{(}\PY{n}{qc}\PY{o}{.}\PY{n}{depth}\PY{p}{(}\PY{p}{)}\PY{p}{)}
\PY{n}{qc} \PY{o}{=} \PY{n}{lgr}\PY{o}{.}\PY{n}{line\PYZus{}graph\PYZus{}route}\PY{p}{(}\PY{n}{qc}\PY{p}{)}
\PY{n+nb}{print}\PY{p}{(}\PY{n}{qc}\PY{o}{.}\PY{n}{depth}\PY{p}{(}\PY{p}{)}\PY{p}{)}
\PY{n}{cg} \PY{o}{=} \PY{n}{lgr}\PY{o}{.}\PY{n}{coupling\PYZus{}graph}\PY{p}{(}\PY{n}{qc}\PY{p}{)}
\PY{n}{nx}\PY{o}{.}\PY{n}{draw\PYZus{}kamada\PYZus{}kawai}\PY{p}{(}\PY{n}{cg}\PY{p}{)}
\end{Verbatim}
\end{tcolorbox}

    \begin{Verbatim}[commandchars=\\\{\}]
4
20
    \end{Verbatim}

    \begin{center}
    \adjustimage{max size={0.9\linewidth}{0.9\paperheight}}{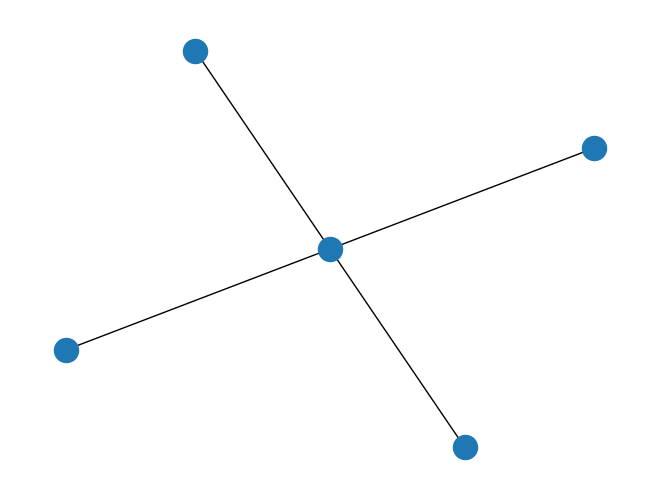}
    \end{center}

    We do not show the circuit diagram in this case because the routed
circuit is not a circuit on a line.

    \hypertarget{shuriken-to-heavy-square-octagon}{%
\section{Shuriken to heavy
square-octagon}\label{shuriken-to-heavy-square-octagon}}

    \hypertarget{random}{%
\subsection{Random}\label{random}}

    Create a random circuit on a patch of the shuriken lattice of
\(3\times 3\) unit cells.

    \begin{tcolorbox}[breakable, size=fbox, boxrule=1pt, pad at break*=1mm,colback=cellbackground, colframe=cellborder]
\prompt{In}{incolor}{13}{\boxspacing}
\begin{Verbatim}[commandchars=\\\{\}]
\PY{n}{n} \PY{o}{=} \PY{l+m+mi}{5}
\PY{n}{lg} \PY{o}{=} \PY{n}{lgr}\PY{o}{.}\PY{n}{shuriken}\PY{p}{(}\PY{l+m+mi}{3}\PY{p}{,} \PY{l+m+mi}{3}\PY{p}{)}
\PY{n}{qc} \PY{o}{=} \PY{n}{lgr}\PY{o}{.}\PY{n}{random\PYZus{}circuit}\PY{p}{(}\PY{n}{lg}\PY{p}{,} \PY{l+m+mi}{10}\PY{o}{*}\PY{o}{*}\PY{l+m+mi}{4}\PY{p}{)}
\PY{n}{cg} \PY{o}{=} \PY{n}{lgr}\PY{o}{.}\PY{n}{coupling\PYZus{}graph}\PY{p}{(}\PY{n}{qc}\PY{p}{)}
\PY{n}{nx}\PY{o}{.}\PY{n}{draw\PYZus{}kamada\PYZus{}kawai}\PY{p}{(}\PY{n}{cg}\PY{p}{)}
\PY{n+nb}{print}\PY{p}{(}\PY{n}{qc}\PY{o}{.}\PY{n}{depth}\PY{p}{(}\PY{p}{)}\PY{p}{)}
\end{Verbatim}
\end{tcolorbox}

    \begin{Verbatim}[commandchars=\\\{\}]
517
    \end{Verbatim}

    \begin{center}
    \adjustimage{max size={0.9\linewidth}{0.9\paperheight}}{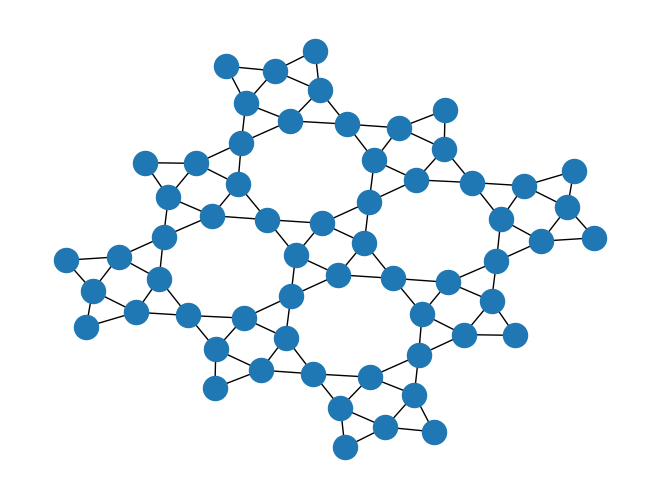}
    \end{center}

    Route the circuit to a circuit with heavy-square-octagon connectivity.

    \begin{tcolorbox}[breakable, size=fbox, boxrule=1pt, pad at break*=1mm,colback=cellbackground, colframe=cellborder]
\prompt{In}{incolor}{14}{\boxspacing}
\begin{Verbatim}[commandchars=\\\{\}]
\PY{n}{qc} \PY{o}{=} \PY{n}{lgr}\PY{o}{.}\PY{n}{line\PYZus{}graph\PYZus{}route}\PY{p}{(}\PY{n}{qc}\PY{p}{)}
\PY{n}{cg} \PY{o}{=} \PY{n}{lgr}\PY{o}{.}\PY{n}{coupling\PYZus{}graph}\PY{p}{(}\PY{n}{qc}\PY{p}{)}
\PY{n}{nx}\PY{o}{.}\PY{n}{draw\PYZus{}kamada\PYZus{}kawai}\PY{p}{(}\PY{n}{cg}\PY{p}{)}
\PY{n+nb}{print}\PY{p}{(}\PY{n}{qc}\PY{o}{.}\PY{n}{depth}\PY{p}{(}\PY{p}{)}\PY{p}{)}
\end{Verbatim}
\end{tcolorbox}

    \begin{Verbatim}[commandchars=\\\{\}]
877
    \end{Verbatim}

    \begin{center}
    \adjustimage{max size={0.9\linewidth}{0.9\paperheight}}{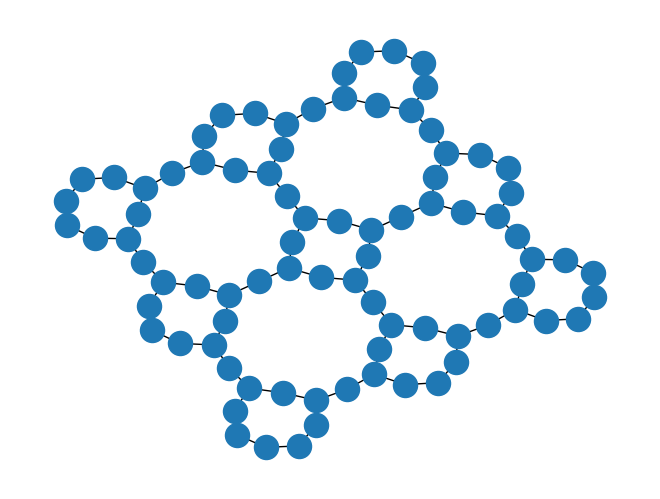}
    \end{center}

    \hypertarget{quantum-simulation}{%
\subsection{Quantum simulation}\label{quantum-simulation}}

As before, circuits are defined by identifying every color with a layer
of \(\mathrm{HEIS}\)-gates. For more details, see the
\texttt{kagome\ to\ heavy-hex} section.

Create and show an edge coloring of the shuriken lattice.

    \begin{tcolorbox}[breakable, size=fbox, boxrule=1pt, pad at break*=1mm,colback=cellbackground, colframe=cellborder]
\prompt{In}{incolor}{15}{\boxspacing}
\begin{Verbatim}[commandchars=\\\{\}]
\PY{n}{lg}  \PY{o}{=}  \PY{n}{lgr}\PY{o}{.}\PY{n}{edge\PYZus{}coloring}\PY{p}{(}\PY{n}{lg}\PY{p}{)}
\PY{n}{lgr}\PY{o}{.}\PY{n}{draw\PYZus{}edge\PYZus{}coloring}\PY{p}{(}\PY{n}{lg}\PY{p}{)}
\end{Verbatim}
\end{tcolorbox}

    \begin{Verbatim}[commandchars=\\\{\}]
Matching is perfect
Edge coloring is not minimal
    \end{Verbatim}

    \begin{center}
    \adjustimage{max size={0.9\linewidth}{0.9\paperheight}}{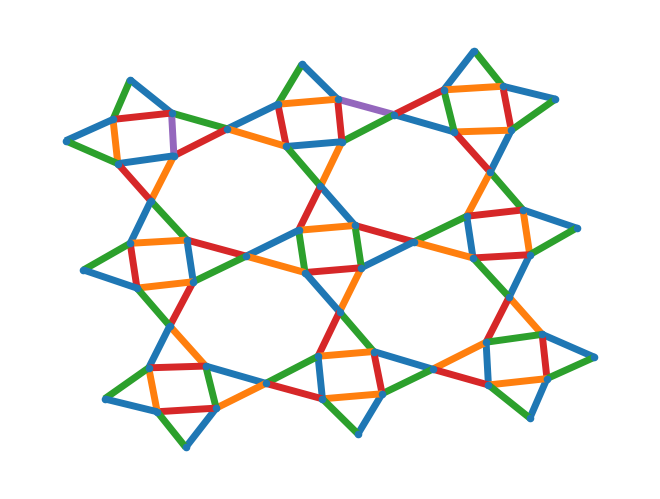}
    \end{center}

    Create the associated circuit, route it to heavy-square-octagon
hardware, and show the coupling graph of the routed circuit.

    \begin{tcolorbox}[breakable, size=fbox, boxrule=1pt, pad at break*=1mm,colback=cellbackground, colframe=cellborder]
\prompt{In}{incolor}{16}{\boxspacing}
\begin{Verbatim}[commandchars=\\\{\}]
\PY{n}{p} \PY{o}{=} \PY{l+m+mi}{1}
\PY{n}{qc} \PY{o}{=} \PY{n}{lgr}\PY{o}{.}\PY{n}{heis\PYZus{}circuit}\PY{p}{(}\PY{n}{lg}\PY{p}{,} \PY{n}{p}\PY{p}{)}
\PY{n+nb}{print}\PY{p}{(}\PY{n}{qc}\PY{o}{.}\PY{n}{depth}\PY{p}{(}\PY{p}{)}\PY{p}{)}
\PY{n}{qc} \PY{o}{=} \PY{n}{lgr}\PY{o}{.}\PY{n}{line\PYZus{}graph\PYZus{}route}\PY{p}{(}\PY{n}{qc}\PY{p}{)}
\PY{n+nb}{print}\PY{p}{(}\PY{n}{qc}\PY{o}{.}\PY{n}{depth}\PY{p}{(}\PY{p}{)}\PY{p}{)}
\PY{n}{cg} \PY{o}{=} \PY{n}{lgr}\PY{o}{.}\PY{n}{coupling\PYZus{}graph}\PY{p}{(}\PY{n}{qc}\PY{p}{)}
\PY{n}{nx}\PY{o}{.}\PY{n}{draw\PYZus{}kamada\PYZus{}kawai}\PY{p}{(}\PY{n}{cg}\PY{p}{)}
\end{Verbatim}
\end{tcolorbox}

    \begin{Verbatim}[commandchars=\\\{\}]
6
12
    \end{Verbatim}

    \begin{center}
    \adjustimage{max size={0.9\linewidth}{0.9\paperheight}}{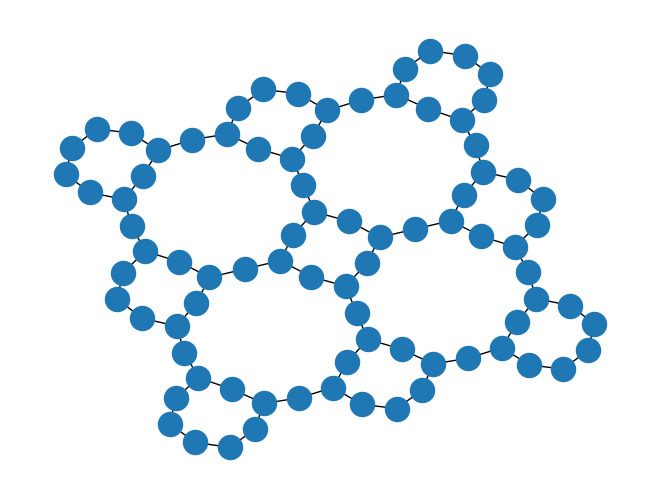}
    \end{center}

    Again, the resulting circuit diagram will not be very insightful, but it
will be for a single-unit cell patch of the shuriken lattice.

    \begin{tcolorbox}[breakable, size=fbox, boxrule=1pt, pad at break*=1mm,colback=cellbackground, colframe=cellborder]
\prompt{In}{incolor}{17}{\boxspacing}
\begin{Verbatim}[commandchars=\\\{\}]
\PY{n}{lg} \PY{o}{=} \PY{n}{lgr}\PY{o}{.}\PY{n}{shuriken}\PY{p}{(}\PY{l+m+mi}{1}\PY{p}{,} \PY{l+m+mi}{1}\PY{p}{)}
\PY{n}{lg} \PY{o}{=} \PY{n}{lgr}\PY{o}{.}\PY{n}{edge\PYZus{}coloring}\PY{p}{(}\PY{n}{lg}\PY{p}{)}
\PY{n}{lgr}\PY{o}{.}\PY{n}{draw\PYZus{}edge\PYZus{}coloring}\PY{p}{(}\PY{n}{lg}\PY{p}{)}
\end{Verbatim}
\end{tcolorbox}

    \begin{Verbatim}[commandchars=\\\{\}]
Matching is perfect
Edge coloring is minimal
    \end{Verbatim}

    \begin{center}
    \adjustimage{max size={0.9\linewidth}{0.9\paperheight}}{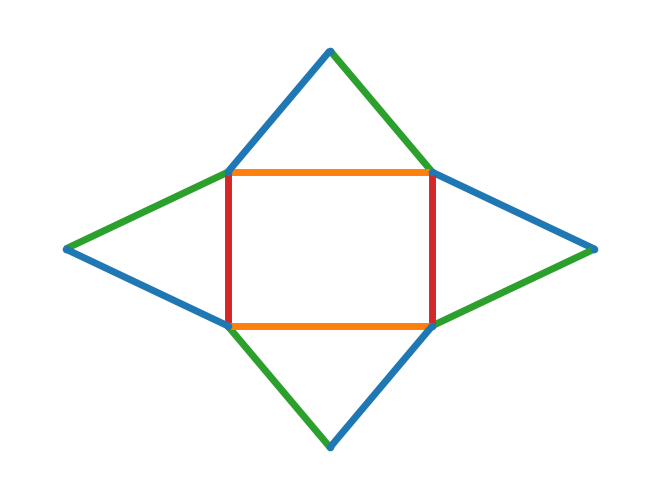}
    \end{center}

    \begin{tcolorbox}[breakable, size=fbox, boxrule=1pt, pad at break*=1mm,colback=cellbackground, colframe=cellborder]
\prompt{In}{incolor}{18}{\boxspacing}
\begin{Verbatim}[commandchars=\\\{\}]
\PY{n}{p} \PY{o}{=} \PY{l+m+mi}{1}
\PY{n}{qc} \PY{o}{=} \PY{n}{lgr}\PY{o}{.}\PY{n}{heis\PYZus{}circuit}\PY{p}{(}\PY{n}{lg}\PY{p}{,} \PY{n}{p}\PY{p}{)}
\PY{n+nb}{print}\PY{p}{(}\PY{n}{qc}\PY{o}{.}\PY{n}{depth}\PY{p}{(}\PY{p}{)}\PY{p}{)}
\PY{n}{qc} \PY{o}{=} \PY{n}{lgr}\PY{o}{.}\PY{n}{line\PYZus{}graph\PYZus{}route}\PY{p}{(}\PY{n}{qc}\PY{p}{)}
\PY{n+nb}{print}\PY{p}{(}\PY{n}{qc}\PY{o}{.}\PY{n}{depth}\PY{p}{(}\PY{p}{)}\PY{p}{)}
\PY{n}{cg} \PY{o}{=} \PY{n}{lgr}\PY{o}{.}\PY{n}{coupling\PYZus{}graph}\PY{p}{(}\PY{n}{qc}\PY{p}{)}
\PY{n}{nx}\PY{o}{.}\PY{n}{draw\PYZus{}kamada\PYZus{}kawai}\PY{p}{(}\PY{n}{cg}\PY{p}{,} \PY{n}{with\PYZus{}labels} \PY{o}{=} \PY{l+s+s1}{\PYZsq{}}\PY{l+s+s1}{true}\PY{l+s+s1}{\PYZsq{}}\PY{p}{)}
\end{Verbatim}
\end{tcolorbox}

    \begin{Verbatim}[commandchars=\\\{\}]
5
9
    \end{Verbatim}

    \begin{center}
    \adjustimage{max size={0.9\linewidth}{0.9\paperheight}}{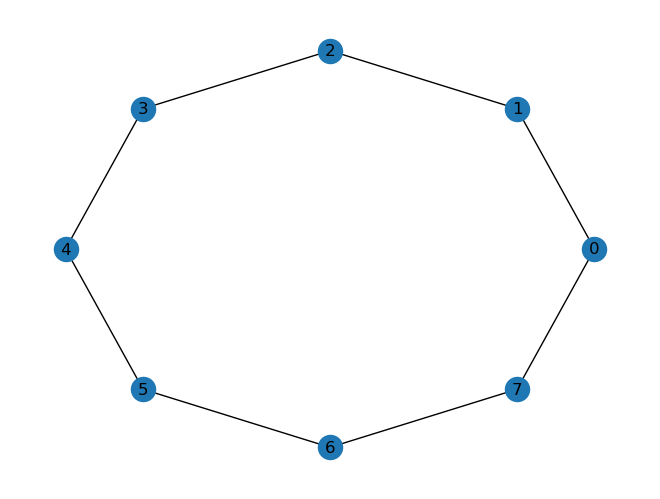}
    \end{center}

    \begin{tcolorbox}[breakable, size=fbox, boxrule=1pt, pad at break*=1mm,colback=cellbackground, colframe=cellborder]
\prompt{In}{incolor}{19}{\boxspacing}
\begin{Verbatim}[commandchars=\\\{\}]
\PY{n}{qc}\PY{o}{.}\PY{n}{draw}\PY{p}{(}\PY{l+s+s1}{\PYZsq{}}\PY{l+s+s1}{latex}\PY{l+s+s1}{\PYZsq{}}\PY{p}{)}
\end{Verbatim}
\end{tcolorbox}

\prompt{Out}{outcolor}{19}{}
    
    \begin{center}
    \adjustimage{max size={0.9\linewidth}{0.9\paperheight}}{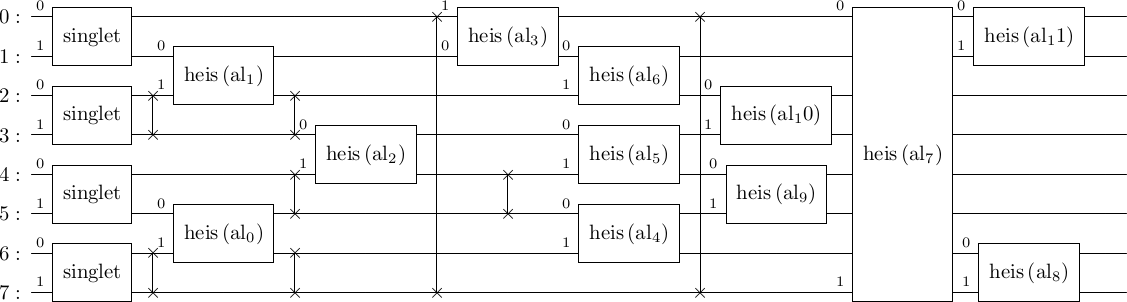}
    \end{center}

    \hypertarget{checkerboard-to-heavy-square}{%
\section{Checkerboard to
heavy-square}\label{checkerboard-to-heavy-square}}

    \begin{tcolorbox}[breakable, size=fbox, boxrule=1pt, pad at break*=1mm,colback=cellbackground, colframe=cellborder]
\prompt{In}{incolor}{20}{\boxspacing}
\begin{Verbatim}[commandchars=\\\{\}]
\PY{n}{m} \PY{o}{=} \PY{l+m+mf}{2.5} \PY{c+c1}{\PYZsh{} For the checkerboard lattice, specify dimentions in nodes by nodes.}
\PY{n}{lg} \PY{o}{=} \PY{n}{lgr}\PY{o}{.}\PY{n}{checkerboard}\PY{p}{(}\PY{n}{m}\PY{p}{,} \PY{n}{m}\PY{p}{)}
\PY{n}{qc} \PY{o}{=} \PY{n}{lgr}\PY{o}{.}\PY{n}{random\PYZus{}circuit}\PY{p}{(}\PY{n}{lg}\PY{p}{,} \PY{l+m+mi}{10}\PY{o}{*}\PY{o}{*}\PY{l+m+mi}{4}\PY{p}{)}
\PY{n}{cg} \PY{o}{=} \PY{n}{lgr}\PY{o}{.}\PY{n}{coupling\PYZus{}graph}\PY{p}{(}\PY{n}{qc}\PY{p}{)}
\PY{n}{nx}\PY{o}{.}\PY{n}{draw\PYZus{}kamada\PYZus{}kawai}\PY{p}{(}\PY{n}{cg}\PY{p}{)}
\PY{n+nb}{print}\PY{p}{(}\PY{n}{qc}\PY{o}{.}\PY{n}{depth}\PY{p}{(}\PY{p}{)}\PY{p}{)}
\end{Verbatim}
\end{tcolorbox}

    \begin{Verbatim}[commandchars=\\\{\}]
872
    \end{Verbatim}

    \begin{center}
    \adjustimage{max size={0.9\linewidth}{0.9\paperheight}}{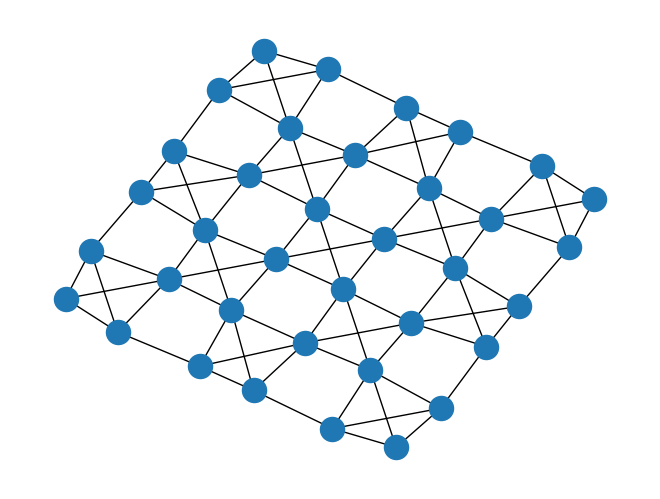}
    \end{center}

    Route the circuit to a circuit with a heavy-square coupling graph.

    \begin{tcolorbox}[breakable, size=fbox, boxrule=1pt, pad at break*=1mm,colback=cellbackground, colframe=cellborder]
\prompt{In}{incolor}{21}{\boxspacing}
\begin{Verbatim}[commandchars=\\\{\}]
\PY{n}{qc} \PY{o}{=} \PY{n}{lgr}\PY{o}{.}\PY{n}{line\PYZus{}graph\PYZus{}route}\PY{p}{(}\PY{n}{qc}\PY{p}{)}
\PY{n}{cg} \PY{o}{=} \PY{n}{lgr}\PY{o}{.}\PY{n}{coupling\PYZus{}graph}\PY{p}{(}\PY{n}{qc}\PY{p}{)}
\PY{n}{nx}\PY{o}{.}\PY{n}{draw\PYZus{}kamada\PYZus{}kawai}\PY{p}{(}\PY{n}{cg}\PY{p}{)}
\PY{n+nb}{print}\PY{p}{(}\PY{n}{qc}\PY{o}{.}\PY{n}{depth}\PY{p}{(}\PY{p}{)}\PY{p}{)}
\end{Verbatim}
\end{tcolorbox}

    \begin{Verbatim}[commandchars=\\\{\}]
1622
    \end{Verbatim}

    \begin{center}
    \adjustimage{max size={0.9\linewidth}{0.9\paperheight}}{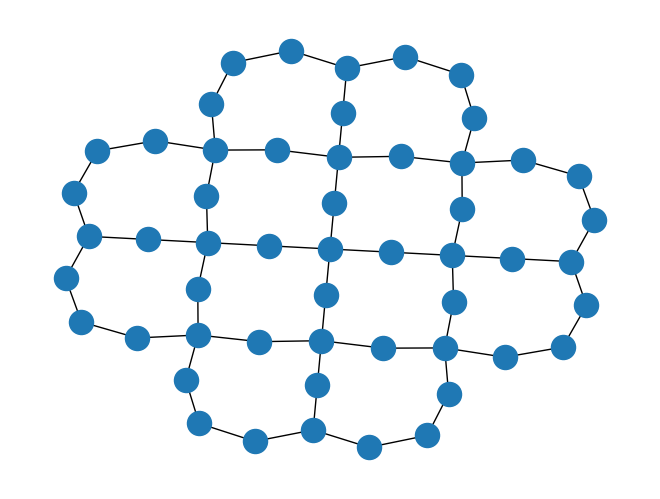}
    \end{center}

    \hypertarget{quantum-simulation}{%
\subsection{Quantum simulation}\label{quantum-simulation}}

As before, circuits are defined by identifying every color with a layer
of \(\mathrm{HEIS}\)-gates. For more details, see the
\texttt{kagome\ to\ heavy-hex} section.

Create and show an edge coloring of the checkerboard lattice

    \begin{tcolorbox}[breakable, size=fbox, boxrule=1pt, pad at break*=1mm,colback=cellbackground, colframe=cellborder]
\prompt{In}{incolor}{22}{\boxspacing}
\begin{Verbatim}[commandchars=\\\{\}]
\PY{n}{lg}  \PY{o}{=}  \PY{n}{lgr}\PY{o}{.}\PY{n}{edge\PYZus{}coloring}\PY{p}{(}\PY{n}{lg}\PY{p}{)}
\PY{n}{lgr}\PY{o}{.}\PY{n}{draw\PYZus{}edge\PYZus{}coloring}\PY{p}{(}\PY{n}{lg}\PY{p}{,} \PY{n}{spectral}\PY{o}{=}\PY{k+kc}{True}\PY{p}{)} \PY{c+c1}{\PYZsh{} Use spactral method to find location of nodes.}
\end{Verbatim}
\end{tcolorbox}

    \begin{Verbatim}[commandchars=\\\{\}]
Matching is perfect
Edge coloring is not minimal
    \end{Verbatim}

    \begin{center}
    \adjustimage{max size={0.9\linewidth}{0.9\paperheight}}{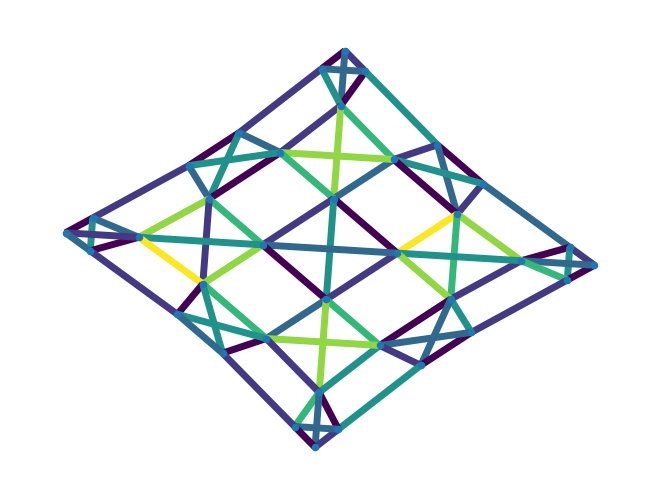}
    \end{center}

    Create the associated circuit, route it to heavy-square-octagon
hardware, and show the coupling graph of the routed circuit.

    \begin{tcolorbox}[breakable, size=fbox, boxrule=1pt, pad at break*=1mm,colback=cellbackground, colframe=cellborder]
\prompt{In}{incolor}{23}{\boxspacing}
\begin{Verbatim}[commandchars=\\\{\}]
\PY{n}{p} \PY{o}{=} \PY{l+m+mi}{1}
\PY{n}{qc} \PY{o}{=} \PY{n}{lgr}\PY{o}{.}\PY{n}{heis\PYZus{}circuit}\PY{p}{(}\PY{n}{lg}\PY{p}{,} \PY{n}{p}\PY{p}{)}
\PY{n+nb}{print}\PY{p}{(}\PY{n}{qc}\PY{o}{.}\PY{n}{depth}\PY{p}{(}\PY{p}{)}\PY{p}{)}
\PY{n}{qc} \PY{o}{=} \PY{n}{lgr}\PY{o}{.}\PY{n}{line\PYZus{}graph\PYZus{}route}\PY{p}{(}\PY{n}{qc}\PY{p}{)}
\PY{n+nb}{print}\PY{p}{(}\PY{n}{qc}\PY{o}{.}\PY{n}{depth}\PY{p}{(}\PY{p}{)}\PY{p}{)}
\PY{n}{cg} \PY{o}{=} \PY{n}{lgr}\PY{o}{.}\PY{n}{coupling\PYZus{}graph}\PY{p}{(}\PY{n}{qc}\PY{p}{)}
\PY{n}{nx}\PY{o}{.}\PY{n}{draw\PYZus{}kamada\PYZus{}kawai}\PY{p}{(}\PY{n}{cg}\PY{p}{)}
\end{Verbatim}
\end{tcolorbox}

    \begin{Verbatim}[commandchars=\\\{\}]
8
33
    \end{Verbatim}

    \begin{center}
    \adjustimage{max size={0.9\linewidth}{0.9\paperheight}}{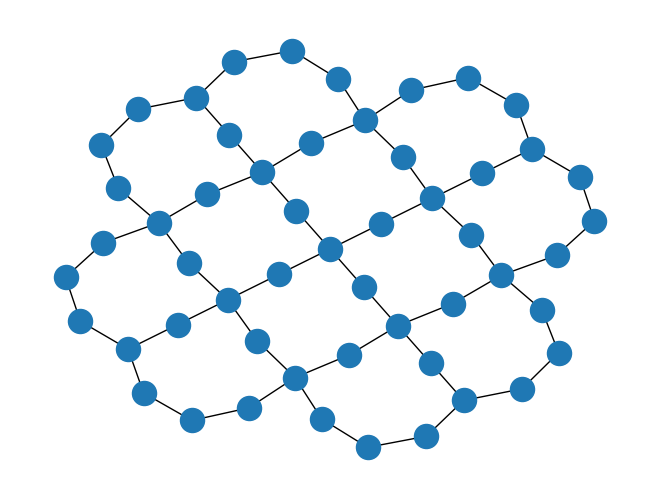}
    \end{center}

    \hypertarget{random-line-graph-to-random-heavy-graph}{%
\section{Random line graph to random heavy
graph}\label{random-line-graph-to-random-heavy-graph}}

    \hypertarget{random}{%
\subsection{Random}\label{random}}

    Create a random circuit on a random graph with 6 nodes and show the
circuit's coupling graph. For details on \texttt{radom\_line\_graph}
generation, see its function definition in
\texttt{line\_graph\_routing.py}.

    \begin{tcolorbox}[breakable, size=fbox, boxrule=1pt, pad at break*=1mm,colback=cellbackground, colframe=cellborder]
\prompt{In}{incolor}{33}{\boxspacing}
\begin{Verbatim}[commandchars=\\\{\}]
\PY{n}{n} \PY{o}{=} \PY{l+m+mi}{6}
\PY{n}{lg} \PY{o}{=} \PY{n}{lgr}\PY{o}{.}\PY{n}{random\PYZus{}line\PYZus{}graph}\PY{p}{(}\PY{l+m+mi}{6}\PY{p}{)}
\PY{n}{qc} \PY{o}{=} \PY{n}{lgr}\PY{o}{.}\PY{n}{random\PYZus{}circuit}\PY{p}{(}\PY{n}{lg}\PY{p}{,} \PY{l+m+mi}{10}\PY{o}{*}\PY{o}{*}\PY{l+m+mi}{3}\PY{p}{)}
\PY{n}{cg} \PY{o}{=} \PY{n}{lgr}\PY{o}{.}\PY{n}{coupling\PYZus{}graph}\PY{p}{(}\PY{n}{qc}\PY{p}{)}
\PY{n}{nx}\PY{o}{.}\PY{n}{draw\PYZus{}kamada\PYZus{}kawai}\PY{p}{(}\PY{n}{cg}\PY{p}{)}
\PY{n+nb}{print}\PY{p}{(}\PY{n}{qc}\PY{o}{.}\PY{n}{depth}\PY{p}{(}\PY{p}{)}\PY{p}{)}
\end{Verbatim}
\end{tcolorbox}

    \begin{Verbatim}[commandchars=\\\{\}]
287
    \end{Verbatim}

    \begin{center}
    \adjustimage{max size={0.9\linewidth}{0.9\paperheight}}{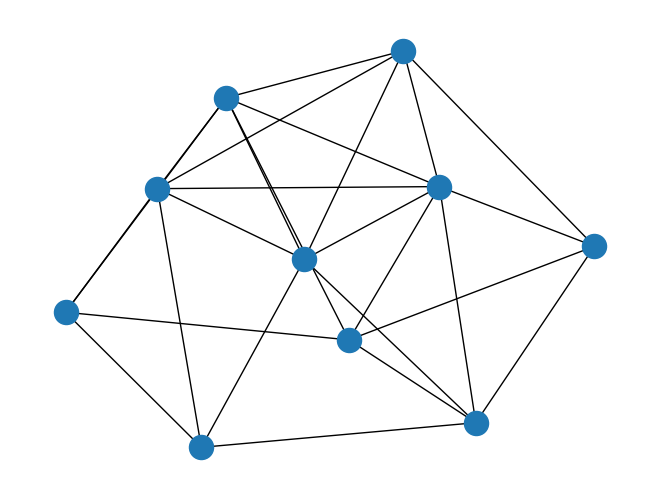}
    \end{center}

    Route the circuit to a circuit with the associated heavy connectivity.

    \begin{tcolorbox}[breakable, size=fbox, boxrule=1pt, pad at break*=1mm,colback=cellbackground, colframe=cellborder]
\prompt{In}{incolor}{34}{\boxspacing}
\begin{Verbatim}[commandchars=\\\{\}]
\PY{n}{qc} \PY{o}{=} \PY{n}{lgr}\PY{o}{.}\PY{n}{line\PYZus{}graph\PYZus{}route}\PY{p}{(}\PY{n}{qc}\PY{p}{)}
\PY{n}{cg} \PY{o}{=} \PY{n}{lgr}\PY{o}{.}\PY{n}{coupling\PYZus{}graph}\PY{p}{(}\PY{n}{qc}\PY{p}{)}
\PY{n}{nx}\PY{o}{.}\PY{n}{draw\PYZus{}kamada\PYZus{}kawai}\PY{p}{(}\PY{n}{cg}\PY{p}{)}
\PY{n+nb}{print}\PY{p}{(}\PY{n}{qc}\PY{o}{.}\PY{n}{depth}\PY{p}{(}\PY{p}{)}\PY{p}{)}
\end{Verbatim}
\end{tcolorbox}

    \begin{Verbatim}[commandchars=\\\{\}]
576
    \end{Verbatim}

    \begin{center}
    \adjustimage{max size={0.9\linewidth}{0.9\paperheight}}{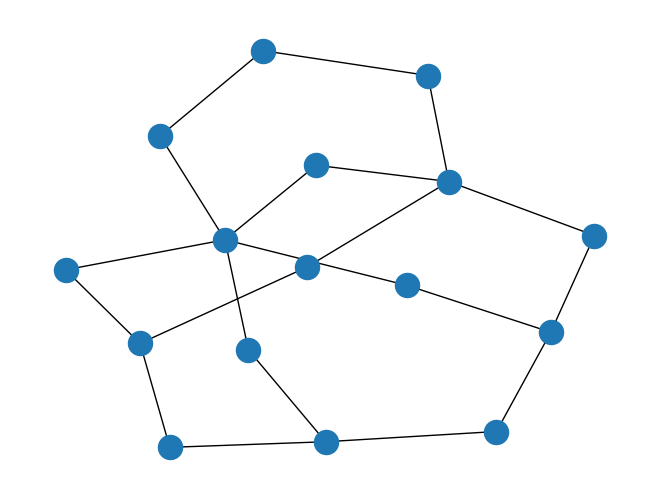}
    \end{center}

    \hypertarget{quantum-simulation}{%
\subsection{Quantum simulation}\label{quantum-simulation}}

As before, circuits are defined by identifying every color with a layer
of \(\mathrm{HEIS}\)-gates. For more details, see the
\texttt{kagome\ to\ heavy-hex} section.

Create and show an edge coloring of random graph. The method we use to
find a perfect matching (needed for initial state preparation) is
limited and may not find a perfect matching even if it exists. If a
perfect matching is not found, try to create another random line graph
(i.e., evaluate the two cells above) or use more sophisticated (or
brute-force) methods to find a perfect matching.

    \begin{tcolorbox}[breakable, size=fbox, boxrule=1pt, pad at break*=1mm,colback=cellbackground, colframe=cellborder]
\prompt{In}{incolor}{35}{\boxspacing}
\begin{Verbatim}[commandchars=\\\{\}]
\PY{n}{lg} \PY{o}{=} \PY{n}{lgr}\PY{o}{.}\PY{n}{edge\PYZus{}coloring}\PY{p}{(}\PY{n}{lg}\PY{p}{)}
\PY{n}{lgr}\PY{o}{.}\PY{n}{draw\PYZus{}edge\PYZus{}coloring}\PY{p}{(}\PY{n}{lg}\PY{p}{)}
\end{Verbatim}
\end{tcolorbox}

    \begin{Verbatim}[commandchars=\\\{\}]
Matching is perfect
Edge coloring is not minimal
    \end{Verbatim}

    \begin{center}
    \adjustimage{max size={0.9\linewidth}{0.9\paperheight}}{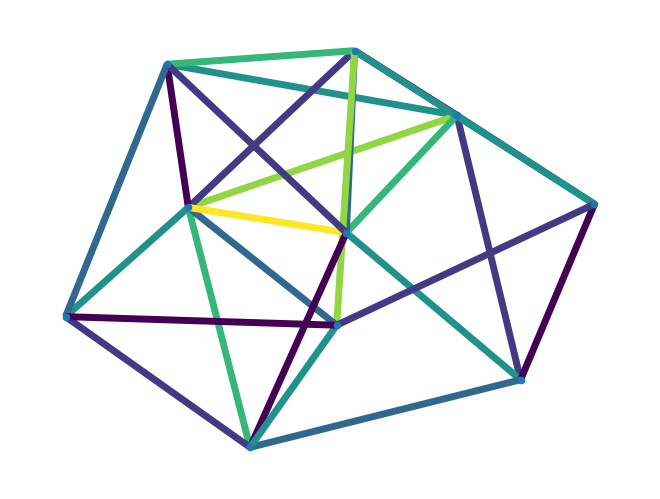}
    \end{center}

    Create the associated circuit, route it to heavy-hex hardware and show
the coupling graph of the routed circuit.

    \begin{tcolorbox}[breakable, size=fbox, boxrule=1pt, pad at break*=1mm,colback=cellbackground, colframe=cellborder]
\prompt{In}{incolor}{36}{\boxspacing}
\begin{Verbatim}[commandchars=\\\{\}]
\PY{n}{p} \PY{o}{=} \PY{l+m+mi}{1}
\PY{n}{qc} \PY{o}{=} \PY{n}{lgr}\PY{o}{.}\PY{n}{heis\PYZus{}circuit}\PY{p}{(}\PY{n}{lg}\PY{p}{,} \PY{n}{p}\PY{p}{)}
\PY{n+nb}{print}\PY{p}{(}\PY{n}{qc}\PY{o}{.}\PY{n}{depth}\PY{p}{(}\PY{p}{)}\PY{p}{)}
\PY{n}{qc} \PY{o}{=} \PY{n}{lgr}\PY{o}{.}\PY{n}{line\PYZus{}graph\PYZus{}route}\PY{p}{(}\PY{n}{qc}\PY{p}{)}
\PY{n+nb}{print}\PY{p}{(}\PY{n}{qc}\PY{o}{.}\PY{n}{depth}\PY{p}{(}\PY{p}{)}\PY{p}{)}
\PY{n}{cg} \PY{o}{=} \PY{n}{lgr}\PY{o}{.}\PY{n}{coupling\PYZus{}graph}\PY{p}{(}\PY{n}{qc}\PY{p}{)}
\PY{n}{nx}\PY{o}{.}\PY{n}{draw\PYZus{}kamada\PYZus{}kawai}\PY{p}{(}\PY{n}{cg}\PY{p}{)}
\end{Verbatim}
\end{tcolorbox}

    \begin{Verbatim}[commandchars=\\\{\}]
8
32
    \end{Verbatim}

    \begin{center}
    \adjustimage{max size={0.9\linewidth}{0.9\paperheight}}{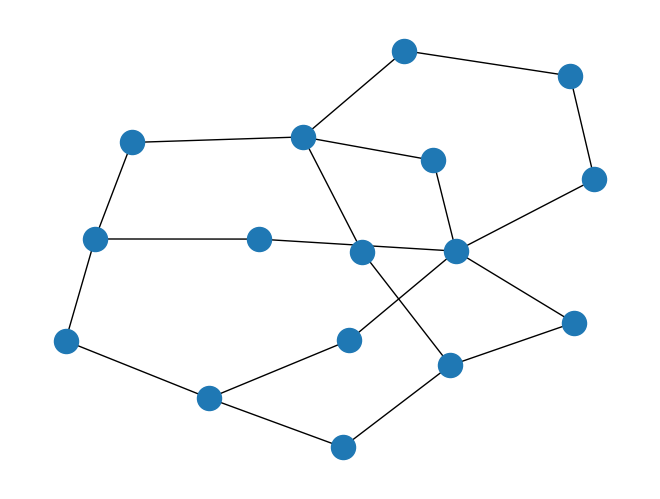}
    \end{center}

    We do not show the circuit diagram in this case because the routed
circuit is not a circuit on a line.

\clearpage
\onecolumn

    \hypertarget{benchmarking}{%
\section{Benchmarking}\label{benchmarking}}

    We benchmark line graph routing by performing the above routing tasks
(but for larger unit cells) using both line-graph qubit routing and all
methods available in qiskit. These methods are `basic', `lookahead',
`stochastic', and `sabre' {[}1{]}.

The benchmarking settings are specified by the following options: -
\texttt{name} The name of the virtual graph, either \texttt{kagome},
\texttt{shuriken} or \texttt{complete}. - \texttt{size}. In case of
\texttt{kagome} and \texttt{shuriken}: the size of the patch in unit
cells by unit cells. In case of \texttt{complete}: the number of nodes
of the complete graph. - \texttt{circuit\ type}. Either
\texttt{quantum\ simulation} or \texttt{random}, as presented in this
notebook. - \texttt{p} In case of \texttt{kagome} and \texttt{shuriken}:
the number of cycles of the circuit. In case of \texttt{complete}: the
number of random gates from the set \texttt{H,T,S,CNOT}. -
\texttt{repetitions}. The number of runs for the methods \texttt{sabre}
and \texttt{stochastic}. The methods \texttt{line-graph} and
\texttt{basic} are deterministic and hence only run once.
Correspondingly, the reported \texttt{total\ time} pertains to the time
taken for this single run in case of the latter two methods. -
\texttt{optimization\_level}. Either 0, 1, 2, or 3. This specifies
the optimization level used for the routing methods implemented in qiskit
{[}2{]}. This parameter is passed directly to Qiskit's transpiler
{[}1{]}. - \texttt{methods}. The methods to benchmark line graph
transpilation against. Must be a list containing elements from
\texttt{{[}\textquotesingle{}sabre\textquotesingle{},\textquotesingle{}basic\textquotesingle{},\textquotesingle{}lookahead\textquotesingle{},stochastic{]}}.
These methods are passed directly to Qiskit's transpiler {[}1{]}.

The methods \texttt{sabre} and \texttt{stochastic} are probabilistic,
achieving a different routing each time they are run, and hence achieve
different performance characteristics with each run. We therefore run
these methods \(repetitions\) times and report the average, confidence
interval, and best performance out of these runs. Error bars on the data
show the (symmetrized) 95\% confidence interval and are obtained by
bootstrapping the data. The error interval for \texttt{num\_qubits} is
sometimes given by \texttt{nan} because in those cases the number of
qubits was equal for all runs. The routing methods\texttt{line-graph}
and \texttt{basic} are deterministic and for these we enforce
\texttt{repetitions=1}.

We consider the following performance characteristics. - \texttt{method}
The routing method. - \texttt{av.\ n\_swaps} The average number of swaps
obtained among the \texttt{repetitions} runs of the routing method. -
\texttt{min\ n\_swap} The number of swaps of the run that achieved the
lowest depth. - \texttt{av.\ depth} The average depth of the routed
circuits among the \texttt{repetitions} runs of the routing method. We
focus on the performance of routing so none of the gates in any of the
routing methods are compiled into hardware native gates. That is, for the
purposes of assessing routing performance, we assume the gate set
\texttt{SWAP}, \texttt{HEIS}, \texttt{H}, \texttt{X}, \texttt{Z}
\texttt{CNOT} for the quantum simulation circuits. For the random
circuits we assume \texttt{SWAP}, \texttt{CNOT}, \texttt{H}, \texttt{S},
\texttt{T}. - \texttt{min\ depth} The minimum depth among the
\texttt{repetitions} runs of the routing method. -
\texttt{av.\ n\_qubits} The average number obtained among the
\texttt{repetitions} runs of the routing method. - \texttt{total\ time}
The total wall clock time (in seconds) needed to perform all
\texttt{repetitions} runs of the routing method. - \texttt{av.\ time}
The average (minimum) wall clock time of the \texttt{repetitions}
routing runs. - \texttt{min.\ time} The number wall clock run time of
the run that achieved the lowest depth.

{[}1{]} Qiskit 0.43.0 documentation,
https://qiskit.org/documentation/stubs/qiskit.compiler.transpile.html,
accessed 11h May 2023.\\ {[}2{]}
https://github.com/Qiskit/qiskit-terra/tree/main/qiskit/transpiler/preset\_passmanagers

    \hypertarget{quantum-simulation-kagome-and-shuriken-agianst-sabre}{%
\subsubsection{Quantum simulation, kagome and shuriken, agianst
SABRE}\label{quantum-simulation-kagome-and-shuriken-agianst-sabre}}

    \begin{tcolorbox}[breakable, size=fbox, boxrule=1pt, pad at break*=1mm,colback=cellbackground, colframe=cellborder]
\prompt{In}{incolor}{37}{\boxspacing}
\begin{Verbatim}[commandchars=\\\{\}]
\PY{k+kn}{import} \PY{n+nn}{line\PYZus{}graph\PYZus{}routing} \PY{k}{as} \PY{n+nn}{lgr} \PY{c+c1}{\PYZsh{} Loading these makes these cells stand\PYZhy{}alone}
\PY{k+kn}{import} \PY{n+nn}{pickle}

\PY{n}{settings}\PY{o}{=}\PY{p}{[}\PY{p}{]}
\PY{k}{for} \PY{n}{name} \PY{o+ow}{in} \PY{p}{[}\PY{l+s+s1}{\PYZsq{}}\PY{l+s+s1}{kagome}\PY{l+s+s1}{\PYZsq{}}\PY{p}{,}\PY{l+s+s1}{\PYZsq{}}\PY{l+s+s1}{shuriken}\PY{l+s+s1}{\PYZsq{}}\PY{p}{]}\PY{p}{:}
    \PY{k}{for} \PY{n}{side} \PY{o+ow}{in} \PY{n+nb}{range}\PY{p}{(}\PY{l+m+mi}{1}\PY{p}{,}\PY{l+m+mi}{9}\PY{p}{,}\PY{l+m+mi}{2}\PY{p}{)}\PY{p}{:}
        \PY{k}{for} \PY{n}{p} \PY{o+ow}{in} \PY{p}{[}\PY{l+m+mi}{1}\PY{p}{,}\PY{l+m+mi}{8}\PY{p}{,}\PY{l+m+mi}{16}\PY{p}{]}\PY{p}{:}
            \PY{k}{for} \PY{n}{optimization\PYZus{}level} \PY{o+ow}{in} \PY{n+nb}{range}\PY{p}{(}\PY{l+m+mi}{4}\PY{p}{)}\PY{p}{:}
                \PY{n}{setting}\PY{o}{=}\PY{p}{\PYZob{}}\PY{l+s+s1}{\PYZsq{}}\PY{l+s+s1}{name}\PY{l+s+s1}{\PYZsq{}}\PY{p}{:}\PY{n}{name}\PY{p}{,}
                     \PY{l+s+s1}{\PYZsq{}}\PY{l+s+s1}{size}\PY{l+s+s1}{\PYZsq{}}\PY{p}{:} \PY{p}{(}\PY{n}{side}\PY{p}{,}\PY{n}{side}\PY{p}{)}\PY{p}{,}
                     \PY{l+s+s1}{\PYZsq{}}\PY{l+s+s1}{circuit\PYZus{}type}\PY{l+s+s1}{\PYZsq{}}\PY{p}{:} \PY{l+s+s1}{\PYZsq{}}\PY{l+s+s1}{quantum\PYZus{}simulation}\PY{l+s+s1}{\PYZsq{}}\PY{p}{,}
                     \PY{l+s+s1}{\PYZsq{}}\PY{l+s+s1}{p}\PY{l+s+s1}{\PYZsq{}}\PY{p}{:} \PY{n}{p}\PY{p}{,}
                     \PY{l+s+s1}{\PYZsq{}}\PY{l+s+s1}{repetitions}\PY{l+s+s1}{\PYZsq{}} \PY{p}{:} \PY{l+m+mi}{16}\PY{p}{,}
                     \PY{l+s+s1}{\PYZsq{}}\PY{l+s+s1}{optimization\PYZus{}level}\PY{l+s+s1}{\PYZsq{}} \PY{p}{:} \PY{n}{optimization\PYZus{}level}\PY{p}{,}
                     \PY{l+s+s1}{\PYZsq{}}\PY{l+s+s1}{methods}\PY{l+s+s1}{\PYZsq{}} \PY{p}{:} \PY{p}{[}\PY{l+s+s1}{\PYZsq{}}\PY{l+s+s1}{sabre}\PY{l+s+s1}{\PYZsq{}}\PY{p}{]}
                    \PY{p}{\PYZcb{}}
                \PY{n}{settings}\PY{o}{.}\PY{n}{append}\PY{p}{(}\PY{n}{setting}\PY{p}{)}
                                
\PY{c+c1}{\PYZsh{}\PYZsh{} Uncomment to rerun benchmarks. This takes a couple of hours.}
\PY{c+c1}{\PYZsh{}results=[]}
\PY{c+c1}{\PYZsh{}for setting in settings:}
\PY{c+c1}{\PYZsh{}    result=lgr.benchmark(**setting)}
\PY{c+c1}{\PYZsh{}    results.append(result)}
\PY{c+c1}{\PYZsh{}    lgr.print\PYZus{}benchmark(result)}
\PY{c+c1}{\PYZsh{}}
\PY{c+c1}{\PYZsh{}with open(\PYZsq{}benchmark\PYZus{}data/kagome\PYZus{}shuriken.pkl\PYZsq{},\PYZsq{}wb\PYZsq{}) as f:}
\PY{c+c1}{\PYZsh{}    pickle.dump(results,f)}

\PY{c+c1}{\PYZsh{}Load previously obtained results from disk and show them. }
\PY{k+kn}{import} \PY{n+nn}{pickle}
\PY{k}{with} \PY{n+nb}{open}\PY{p}{(}\PY{l+s+s1}{\PYZsq{}}\PY{l+s+s1}{benchmark\PYZus{}data/kagome\PYZus{}shuriken.pkl}\PY{l+s+s1}{\PYZsq{}}\PY{p}{,}\PY{l+s+s1}{\PYZsq{}}\PY{l+s+s1}{rb}\PY{l+s+s1}{\PYZsq{}}\PY{p}{)} \PY{k}{as} \PY{n}{f}\PY{p}{:}
    \PY{n}{results}\PY{o}{=}\PY{n}{pickle}\PY{o}{.}\PY{n}{load}\PY{p}{(}\PY{n}{f}\PY{p}{)}

\PY{k}{for} \PY{n}{result} \PY{o+ow}{in} \PY{n}{results}\PY{p}{:}
    \PY{n}{lgr}\PY{o}{.}\PY{n}{print\PYZus{}benchmark}\PY{p}{(}\PY{n}{result}\PY{p}{)}
\end{Verbatim}
\end{tcolorbox}



    \hypertarget{quantum-simulation-checkerboard-agianst-sabre}{%
\subsubsection{Quantum simulation, checkerboard, agianst
SABRE}\label{quantum-simulation-checkerboard-agianst-sabre}}

    \begin{tcolorbox}[breakable, size=fbox, boxrule=1pt, pad at break*=1mm,colback=cellbackground, colframe=cellborder]
\prompt{In}{incolor}{39}{\boxspacing}
\begin{Verbatim}[commandchars=\\\{\}]
\PY{k+kn}{import} \PY{n+nn}{line\PYZus{}graph\PYZus{}routing} \PY{k}{as} \PY{n+nn}{lgr} \PY{c+c1}{\PYZsh{} Loading these makes these cells stand\PYZhy{}alone}
\PY{k+kn}{import} \PY{n+nn}{pickle}

\PY{n}{settings}\PY{o}{=}\PY{p}{[}\PY{p}{]}
\PY{k}{for} \PY{n}{name} \PY{o+ow}{in} \PY{p}{[}\PY{l+s+s1}{\PYZsq{}}\PY{l+s+s1}{checkerboard}\PY{l+s+s1}{\PYZsq{}}\PY{p}{]}\PY{p}{:}
    \PY{k}{for} \PY{n}{side} \PY{o+ow}{in} \PY{p}{[}\PY{n}{i}\PY{o}{+}\PY{l+m+mf}{0.5} \PY{k}{for} \PY{n}{i} \PY{o+ow}{in} \PY{n+nb}{range}\PY{p}{(}\PY{l+m+mi}{1}\PY{p}{,}\PY{l+m+mi}{9}\PY{p}{,}\PY{l+m+mi}{2}\PY{p}{)}\PY{p}{]}\PY{p}{:}
        \PY{k}{for} \PY{n}{p} \PY{o+ow}{in} \PY{p}{[}\PY{l+m+mi}{1}\PY{p}{,}\PY{l+m+mi}{8}\PY{p}{,}\PY{l+m+mi}{16}\PY{p}{]}\PY{p}{:}
            \PY{k}{for} \PY{n}{optimization\PYZus{}level} \PY{o+ow}{in} \PY{n+nb}{range}\PY{p}{(}\PY{l+m+mi}{4}\PY{p}{)}\PY{p}{:}
                \PY{n}{setting}\PY{o}{=}\PY{p}{\PYZob{}}\PY{l+s+s1}{\PYZsq{}}\PY{l+s+s1}{name}\PY{l+s+s1}{\PYZsq{}}\PY{p}{:}\PY{n}{name}\PY{p}{,}
                     \PY{l+s+s1}{\PYZsq{}}\PY{l+s+s1}{size}\PY{l+s+s1}{\PYZsq{}}\PY{p}{:} \PY{p}{(}\PY{n}{side}\PY{p}{,}\PY{n}{side}\PY{p}{)}\PY{p}{,}
                     \PY{l+s+s1}{\PYZsq{}}\PY{l+s+s1}{circuit\PYZus{}type}\PY{l+s+s1}{\PYZsq{}}\PY{p}{:} \PY{l+s+s1}{\PYZsq{}}\PY{l+s+s1}{quantum\PYZus{}simulation}\PY{l+s+s1}{\PYZsq{}}\PY{p}{,}
                     \PY{l+s+s1}{\PYZsq{}}\PY{l+s+s1}{p}\PY{l+s+s1}{\PYZsq{}}\PY{p}{:} \PY{n}{p}\PY{p}{,}
                     \PY{l+s+s1}{\PYZsq{}}\PY{l+s+s1}{repetitions}\PY{l+s+s1}{\PYZsq{}} \PY{p}{:} \PY{l+m+mi}{16}\PY{p}{,}
                     \PY{l+s+s1}{\PYZsq{}}\PY{l+s+s1}{optimization\PYZus{}level}\PY{l+s+s1}{\PYZsq{}} \PY{p}{:} \PY{n}{optimization\PYZus{}level}\PY{p}{,}
                     \PY{l+s+s1}{\PYZsq{}}\PY{l+s+s1}{methods}\PY{l+s+s1}{\PYZsq{}} \PY{p}{:} \PY{p}{[}\PY{l+s+s1}{\PYZsq{}}\PY{l+s+s1}{sabre}\PY{l+s+s1}{\PYZsq{}}\PY{p}{]}
                    \PY{p}{\PYZcb{}}
                \PY{n}{settings}\PY{o}{.}\PY{n}{append}\PY{p}{(}\PY{n}{setting}\PY{p}{)}
                                
\PY{c+c1}{\PYZsh{}\PYZsh{} Uncomment to rerun benchmarks. This takes a couple of hours.}
\PY{c+c1}{\PYZsh{}results=[]}
\PY{c+c1}{\PYZsh{}for setting in settings:}
\PY{c+c1}{\PYZsh{}    result=lgr.benchmark(**setting)}
\PY{c+c1}{\PYZsh{}    results.append(result)}
\PY{c+c1}{\PYZsh{}    lgr.print\PYZus{}benchmark(result)}
\PY{c+c1}{\PYZsh{}}
\PY{c+c1}{\PYZsh{}with open(\PYZsq{}benchmark\PYZus{}data/checkerboard.pkl\PYZsq{},\PYZsq{}wb\PYZsq{}) as f:}
\PY{c+c1}{\PYZsh{}    pickle.dump(results,f)}

\PY{c+c1}{\PYZsh{}Load previously obtained results from disk and show them. }
\PY{k+kn}{import} \PY{n+nn}{pickle}
\PY{k}{with} \PY{n+nb}{open}\PY{p}{(}\PY{l+s+s1}{\PYZsq{}}\PY{l+s+s1}{benchmark\PYZus{}data/checkerboard.pkl}\PY{l+s+s1}{\PYZsq{}}\PY{p}{,}\PY{l+s+s1}{\PYZsq{}}\PY{l+s+s1}{rb}\PY{l+s+s1}{\PYZsq{}}\PY{p}{)} \PY{k}{as} \PY{n}{f}\PY{p}{:}
    \PY{n}{results}\PY{o}{=}\PY{n}{pickle}\PY{o}{.}\PY{n}{load}\PY{p}{(}\PY{n}{f}\PY{p}{)}

\PY{k}{for} \PY{n}{result} \PY{o+ow}{in} \PY{n}{results}\PY{p}{:}
    \PY{n}{lgr}\PY{o}{.}\PY{n}{print\PYZus{}benchmark}\PY{p}{(}\PY{n}{result}\PY{p}{)}
\end{Verbatim}
\end{tcolorbox}



    \hypertarget{random-circuit-kagome-and-shuriken-against-sabre}{%
\subsubsection{Random circuit, kagome and shuriken, against
SABRE}\label{random-circuit-kagome-and-shuriken-against-sabre}}

    \begin{tcolorbox}[breakable, size=fbox, boxrule=1pt, pad at break*=1mm,colback=cellbackground, colframe=cellborder]
\prompt{In}{incolor}{40}{\boxspacing}
\begin{Verbatim}[commandchars=\\\{\}]
\PY{k+kn}{import} \PY{n+nn}{line\PYZus{}graph\PYZus{}routing} \PY{k}{as} \PY{n+nn}{lgr} \PY{c+c1}{\PYZsh{} Loading these makes these cells stand\PYZhy{}alone}
\PY{k+kn}{import} \PY{n+nn}{pickle}

\PY{n}{settings}\PY{o}{=}\PY{p}{[}\PY{p}{]}
\PY{k}{for} \PY{n}{name} \PY{o+ow}{in} \PY{p}{[}\PY{l+s+s1}{\PYZsq{}}\PY{l+s+s1}{kagome}\PY{l+s+s1}{\PYZsq{}}\PY{p}{,}\PY{l+s+s1}{\PYZsq{}}\PY{l+s+s1}{shuriken}\PY{l+s+s1}{\PYZsq{}}\PY{p}{]}\PY{p}{:}
    \PY{k}{for} \PY{n}{side} \PY{o+ow}{in} \PY{n+nb}{range}\PY{p}{(}\PY{l+m+mi}{1}\PY{p}{,}\PY{l+m+mi}{7}\PY{p}{,}\PY{l+m+mi}{2}\PY{p}{)}\PY{p}{:}
        \PY{k}{for} \PY{n}{p} \PY{o+ow}{in} \PY{p}{[}\PY{n}{side}\PY{o}{*}\PY{o}{*}\PY{l+m+mi}{2}\PY{o}{*}\PY{l+m+mi}{500}\PY{p}{]}\PY{p}{:}
            \PY{k}{for} \PY{n}{optimization\PYZus{}level} \PY{o+ow}{in} \PY{p}{[}\PY{l+m+mi}{1}\PY{p}{]}\PY{p}{:}
                \PY{n}{setting}\PY{o}{=}\PY{p}{\PYZob{}}\PY{l+s+s1}{\PYZsq{}}\PY{l+s+s1}{name}\PY{l+s+s1}{\PYZsq{}}\PY{p}{:}\PY{n}{name}\PY{p}{,}
                     \PY{l+s+s1}{\PYZsq{}}\PY{l+s+s1}{size}\PY{l+s+s1}{\PYZsq{}}\PY{p}{:} \PY{p}{(}\PY{n}{side}\PY{p}{,}\PY{n}{side}\PY{p}{)}\PY{p}{,}
                     \PY{l+s+s1}{\PYZsq{}}\PY{l+s+s1}{circuit\PYZus{}type}\PY{l+s+s1}{\PYZsq{}}\PY{p}{:} \PY{l+s+s1}{\PYZsq{}}\PY{l+s+s1}{random}\PY{l+s+s1}{\PYZsq{}}\PY{p}{,}
                     \PY{l+s+s1}{\PYZsq{}}\PY{l+s+s1}{p}\PY{l+s+s1}{\PYZsq{}}\PY{p}{:} \PY{n}{p}\PY{p}{,}
                     \PY{l+s+s1}{\PYZsq{}}\PY{l+s+s1}{repetitions}\PY{l+s+s1}{\PYZsq{}} \PY{p}{:} \PY{l+m+mi}{16}\PY{p}{,}
                     \PY{l+s+s1}{\PYZsq{}}\PY{l+s+s1}{optimization\PYZus{}level}\PY{l+s+s1}{\PYZsq{}} \PY{p}{:} \PY{n}{optimization\PYZus{}level}\PY{p}{,}
                     \PY{l+s+s1}{\PYZsq{}}\PY{l+s+s1}{methods}\PY{l+s+s1}{\PYZsq{}} \PY{p}{:} \PY{p}{[}\PY{l+s+s1}{\PYZsq{}}\PY{l+s+s1}{sabre}\PY{l+s+s1}{\PYZsq{}}\PY{p}{]}
                    \PY{p}{\PYZcb{}}
                \PY{n}{settings}\PY{o}{.}\PY{n}{append}\PY{p}{(}\PY{n}{setting}\PY{p}{)}
                                
\PY{c+c1}{\PYZsh{}\PYZsh{} Uncomment to rerun benchmarks.}
\PY{c+c1}{\PYZsh{}results=[]}
\PY{c+c1}{\PYZsh{}for setting in settings:}
\PY{c+c1}{\PYZsh{}    result=lgr.benchmark(**setting)}
\PY{c+c1}{\PYZsh{}    results.append(result)}
\PY{c+c1}{\PYZsh{}    lgr.print\PYZus{}benchmark(result)}
\PY{c+c1}{\PYZsh{}}
\PY{c+c1}{\PYZsh{}with open(\PYZsq{}benchmark\PYZus{}data/random.pkl\PYZsq{},\PYZsq{}wb\PYZsq{}) as f:}
\PY{c+c1}{\PYZsh{}    pickle.dump(results,f)}

\PY{c+c1}{\PYZsh{}Load previously obtained results from disk and show them. }
\PY{k+kn}{import} \PY{n+nn}{pickle}
\PY{k}{with} \PY{n+nb}{open}\PY{p}{(}\PY{l+s+s1}{\PYZsq{}}\PY{l+s+s1}{benchmark\PYZus{}data/random.pkl}\PY{l+s+s1}{\PYZsq{}}\PY{p}{,}\PY{l+s+s1}{\PYZsq{}}\PY{l+s+s1}{rb}\PY{l+s+s1}{\PYZsq{}}\PY{p}{)} \PY{k}{as} \PY{n}{f}\PY{p}{:}
    \PY{n}{results}\PY{o}{=}\PY{n}{pickle}\PY{o}{.}\PY{n}{load}\PY{p}{(}\PY{n}{f}\PY{p}{)}

\PY{k}{for} \PY{n}{result} \PY{o+ow}{in} \PY{n}{results}\PY{p}{:}
    \PY{n}{lgr}\PY{o}{.}\PY{n}{print\PYZus{}benchmark}\PY{p}{(}\PY{n}{result}\PY{p}{)}
\end{Verbatim}
\end{tcolorbox}

    \begin{Verbatim}[commandchars=\\\{\}]
{\tiny
------------------------------------------------------------------------------------------------------------------------------------------------------
name = kagome, size = (1, 1), circuit_type = random, p = 500, repetitions = 16, optimization_level = 1

method      av. n_swaps      min. n_swap  av. depth       min. depth  av. n_qubits      min. qubits    total time (s)  av. time (s)      min. time (s)
----------  -------------  -------------  ------------  ------------  --------------  -------------  ----------------  --------------  ---------------
line-graph  200 ± 0.0                200  249 ± 0.0              249  12 ± 0.0                   12              0.37  0.37 ± 0.0                 0.37
sabre       96 ± 0.62                 96  207.5 ± 1.19           205  12.0 ± nan                 12              1.17  0.07 ± 0.03                0.06
------------------------------------------------------------------------------------------------------------------------------------------------------
}
{\tiny
------------------------------------------------------------------------------------------------------------------------------------------------------
name = kagome, size = (3, 3), circuit_type = random, p = 4500, repetitions = 16, optimization_level = 1

method      av. n_swaps      min. n_swap  av. depth        min. depth  av. n_qubits      min. qubits    total time (s)  av. time (s)      min. time (s)
----------  -------------  -------------  -------------  ------------  --------------  -------------  ----------------  --------------  ---------------
line-graph  2611 ± 0.0              2611  568 ± 0.0               568  68 ± 0.0                   68              1.14  1.14 ± 0.0                 1.14
sabre       1444 ± 25.94            1450  838.88 ± 35.0           720  44.31 ± 1.75               54             10.24  0.64 ± 0.03                0.66
------------------------------------------------------------------------------------------------------------------------------------------------------
}
{\tiny
------------------------------------------------------------------------------------------------------------------------------------------------------
name = kagome, size = (5, 5), circuit_type = random, p = 12500, repetitions = 16, optimization_level = 1

method      av. n_swaps      min. n_swap  av. depth          min. depth  av. n_qubits      min. qubits    total time (s)  av. time (s)      min. time (s)
----------  -------------  -------------  ---------------  ------------  --------------  -------------  ----------------  --------------  ---------------
line-graph  7592 ± 0.0              7592  735 ± 0.0                 735  164 ± 0.0                 164              3.02  3.02 ± 0.0                 3.02
sabre       5637 ± 82.31            5452  1774.44 ± 54.79          1590  108.69 ± 3.56             107             32.72  2.04 ± 0.04                2.03
------------------------------------------------------------------------------------------------------------------------------------------------------
}
{\tiny
------------------------------------------------------------------------------------------------------------------------------------------------------
name = shuriken, size = (1, 1), circuit_type = random, p = 500, repetitions = 16, optimization_level = 1

method      av. n_swaps      min. n_swap  av. depth        min. depth  av. n_qubits      min. qubits    total time (s)  av. time (s)      min. time (s)
----------  -------------  -------------  -------------  ------------  --------------  -------------  ----------------  --------------  ---------------
line-graph  122 ± 0.0                122  247 ± 0.0               247  8 ± 0.0                     8              0.07  0.07 ± 0.0                 0.07
sabre       85 ± nan                  85  231.12 ± 0.31           231  8.0 ± nan                   8              0.84  0.05 ± 0.0                 0.05
------------------------------------------------------------------------------------------------------------------------------------------------------
}
{\tiny
------------------------------------------------------------------------------------------------------------------------------------------------------
name = shuriken, size = (3, 3), circuit_type = random, p = 4500, repetitions = 16, optimization_level = 1

method      av. n_swaps      min. n_swap  av. depth         min. depth  av. n_qubits      min. qubits    total time (s)  av. time (s)      min. time (s)
----------  -------------  -------------  --------------  ------------  --------------  -------------  ----------------  --------------  ---------------
line-graph  2262 ± 0.0              2262  415 ± 0.0                415  84 ± 0.0                   84              1.05  1.05 ± 0.0                 1.05
sabre       1662 ± 39.26            1447  667.12 ± 39.54           517  65.25 ± 1.38               69             10.53  0.66 ± 0.04                0.7
------------------------------------------------------------------------------------------------------------------------------------------------------
}
{\tiny
------------------------------------------------------------------------------------------------------------------------------------------------------
name = shuriken, size = (5, 5), circuit_type = random, p = 12500, repetitions = 16, optimization_level = 1

method      av. n_swaps      min. n_swap  av. depth          min. depth  av. n_qubits      min. qubits    total time (s)  av. time (s)      min. time (s)
----------  -------------  -------------  ---------------  ------------  --------------  -------------  ----------------  --------------  ---------------
line-graph  6828 ± 0.0              6828  405 ± 0.0                 405  240 ± 0.0                 240              2.83  2.83 ± 0.0                 2.83
sabre       6811 ± 192.15           5969  1403.31 ± 94.74          1094  177.31 ± 2.94             176             38.36  2.4 ± 0.04                 2.38
------------------------------------------------------------------------------------------------------------------------------------------------------
}

    \end{Verbatim}

    \hypertarget{random-circuit-complete-graph-against-sabre}{%
\subsubsection{Random circuit, complete graph, against
SABRE}\label{random-circuit-complete-graph-against-sabre}}

    \begin{tcolorbox}[breakable, size=fbox, boxrule=1pt, pad at break*=1mm,colback=cellbackground, colframe=cellborder]
\prompt{In}{incolor}{41}{\boxspacing}
\begin{Verbatim}[commandchars=\\\{\}]
\PY{k+kn}{import} \PY{n+nn}{line\PYZus{}graph\PYZus{}routing} \PY{k}{as} \PY{n+nn}{lgr}
\PY{k+kn}{import} \PY{n+nn}{pickle}

\PY{n}{settings}\PY{o}{=}\PY{p}{[}\PY{p}{]}
\PY{k}{for} \PY{n}{name} \PY{o+ow}{in} \PY{p}{[}\PY{l+s+s1}{\PYZsq{}}\PY{l+s+s1}{complete}\PY{l+s+s1}{\PYZsq{}}\PY{p}{]}\PY{p}{:}
    \PY{k}{for} \PY{n}{side} \PY{o+ow}{in} \PY{n+nb}{range}\PY{p}{(}\PY{l+m+mi}{3}\PY{p}{,}\PY{l+m+mi}{10}\PY{p}{,}\PY{l+m+mi}{2}\PY{p}{)}\PY{p}{:}
        \PY{k}{for} \PY{n}{p} \PY{o+ow}{in} \PY{p}{[}\PY{n}{side}\PY{o}{*}\PY{l+m+mi}{100}\PY{p}{]}\PY{p}{:}
            \PY{k}{for} \PY{n}{optimization\PYZus{}level} \PY{o+ow}{in} \PY{n+nb}{range}\PY{p}{(}\PY{l+m+mi}{2}\PY{p}{)}\PY{p}{:}
                \PY{n}{setting}\PY{o}{=}\PY{p}{\PYZob{}}\PY{l+s+s1}{\PYZsq{}}\PY{l+s+s1}{name}\PY{l+s+s1}{\PYZsq{}}\PY{p}{:}\PY{n}{name}\PY{p}{,}
                         \PY{l+s+s1}{\PYZsq{}}\PY{l+s+s1}{size}\PY{l+s+s1}{\PYZsq{}}\PY{p}{:} \PY{n}{side}\PY{p}{,}
                         \PY{l+s+s1}{\PYZsq{}}\PY{l+s+s1}{circuit\PYZus{}type}\PY{l+s+s1}{\PYZsq{}}\PY{p}{:} \PY{l+s+s1}{\PYZsq{}}\PY{l+s+s1}{random}\PY{l+s+s1}{\PYZsq{}}\PY{p}{,}
                         \PY{l+s+s1}{\PYZsq{}}\PY{l+s+s1}{p}\PY{l+s+s1}{\PYZsq{}}\PY{p}{:} \PY{n}{p}\PY{p}{,}
                         \PY{l+s+s1}{\PYZsq{}}\PY{l+s+s1}{repetitions}\PY{l+s+s1}{\PYZsq{}} \PY{p}{:} \PY{l+m+mi}{16}\PY{p}{,}
                         \PY{l+s+s1}{\PYZsq{}}\PY{l+s+s1}{optimization\PYZus{}level}\PY{l+s+s1}{\PYZsq{}} \PY{p}{:} \PY{n}{optimization\PYZus{}level}\PY{p}{,}
                         \PY{l+s+s1}{\PYZsq{}}\PY{l+s+s1}{methods}\PY{l+s+s1}{\PYZsq{}} \PY{p}{:} \PY{p}{[}\PY{l+s+s1}{\PYZsq{}}\PY{l+s+s1}{sabre}\PY{l+s+s1}{\PYZsq{}}\PY{p}{]}
                    \PY{p}{\PYZcb{}}
                \PY{n}{settings}\PY{o}{.}\PY{n}{append}\PY{p}{(}\PY{n}{setting}\PY{p}{)}

\PY{c+c1}{\PYZsh{}\PYZsh{} Uncomment to rerun benchmarks}
\PY{c+c1}{\PYZsh{}results=[]}
\PY{c+c1}{\PYZsh{}for setting in settings:}
\PY{c+c1}{\PYZsh{}    result=lgr.benchmark(**setting)}
\PY{c+c1}{\PYZsh{}    results.append(result)}
\PY{c+c1}{\PYZsh{}    lgr.print\PYZus{}benchmark(result)}

\PY{c+c1}{\PYZsh{}with open(\PYZsq{}benchmark\PYZus{}data/complete.pkl\PYZsq{},\PYZsq{}wb\PYZsq{}) as f:}
\PY{c+c1}{\PYZsh{}    pickle.dump(results,f)}

\PY{k}{with} \PY{n+nb}{open}\PY{p}{(}\PY{l+s+s1}{\PYZsq{}}\PY{l+s+s1}{benchmark\PYZus{}data/complete.pkl}\PY{l+s+s1}{\PYZsq{}}\PY{p}{,}\PY{l+s+s1}{\PYZsq{}}\PY{l+s+s1}{rb}\PY{l+s+s1}{\PYZsq{}}\PY{p}{)} \PY{k}{as} \PY{n}{f}\PY{p}{:}
    \PY{n}{results}\PY{o}{=}\PY{n}{pickle}\PY{o}{.}\PY{n}{load}\PY{p}{(}\PY{n}{f}\PY{p}{)}

\PY{k}{for} \PY{n}{result} \PY{o+ow}{in} \PY{n}{results}\PY{p}{:}
    \PY{n}{lgr}\PY{o}{.}\PY{n}{print\PYZus{}benchmark}\PY{p}{(}\PY{n}{result}\PY{p}{)}
\end{Verbatim}
\end{tcolorbox}

    \begin{Verbatim}[commandchars=\\\{\}]
{\tiny
------------------------------------------------------------------------------------------------------------------------------------------------------
name = complete, size = 3, circuit_type = random, p = 300, repetitions = 16, optimization_level = 0

method      av. n_swaps      min. n_swap  av. depth        min. depth  av. n_qubits      min. qubits    total time (s)  av. time (s)      min. time (s)
----------  -------------  -------------  -------------  ------------  --------------  -------------  ----------------  --------------  ---------------
line-graph  154 ± 0.0                154  334 ± 0.0               334  4 ± 0.0                     4              0.12  0.12 ± 0.0                 0.12
sabre       40 ± 0.97                 37  267.88 ± 2.34           257  3.0 ± nan                   3              1.01  0.06 ± 0.06                0.04
------------------------------------------------------------------------------------------------------------------------------------------------------
}
{\tiny
------------------------------------------------------------------------------------------------------------------------------------------------------
name = complete, size = 3, circuit_type = random, p = 300, repetitions = 16, optimization_level = 1

method      av. n_swaps      min. n_swap  av. depth      min. depth  av. n_qubits      min. qubits    total time (s)  av. time (s)      min. time (s)
----------  -------------  -------------  -----------  ------------  --------------  -------------  ----------------  --------------  ---------------
line-graph  140 ± 0.0                140  313 ± 0.0             313  4 ± 0.0                     4              0.11  0.11 ± 0.0                 0.11
sabre       34 ± nan                  34  220.0 ± nan           220  3.0 ± nan                   3              0.97  0.06 ± 0.0                 0.07
------------------------------------------------------------------------------------------------------------------------------------------------------
}
{\tiny
------------------------------------------------------------------------------------------------------------------------------------------------------
name = complete, size = 5, circuit_type = random, p = 500, repetitions = 16, optimization_level = 0

method      av. n_swaps      min. n_swap  av. depth       min. depth  av. n_qubits      min. qubits    total time (s)  av. time (s)      min. time (s)
----------  -------------  -------------  ------------  ------------  --------------  -------------  ----------------  --------------  ---------------
line-graph  314 ± 0.0                314  548 ± 0.0              548  6 ± 0.0                     6              0.19  0.19 ± 0.0                 0.19
sabre       88 ± 1.25                 84  458.06 ± 4.0           444  5.0 ± nan                   5              0.96  0.06 ± nan                 0.06
------------------------------------------------------------------------------------------------------------------------------------------------------
}
{\tiny
------------------------------------------------------------------------------------------------------------------------------------------------------
name = complete, size = 5, circuit_type = random, p = 500, repetitions = 16, optimization_level = 1

method      av. n_swaps      min. n_swap  av. depth        min. depth  av. n_qubits      min. qubits    total time (s)  av. time (s)      min. time (s)
----------  -------------  -------------  -------------  ------------  --------------  -------------  ----------------  --------------  ---------------
line-graph  286 ± 0.0                286  503 ± 0.0               503  6 ± 0.0                     6              0.18  0.18 ± 0.0                 0.18
sabre       72 ± nan                  72  387.25 ± 0.31           386  5.0 ± nan                   5              1.44  0.09 ± nan                 0.09
------------------------------------------------------------------------------------------------------------------------------------------------------
}
{\tiny
------------------------------------------------------------------------------------------------------------------------------------------------------
name = complete, size = 7, circuit_type = random, p = 700, repetitions = 16, optimization_level = 0

method      av. n_swaps      min. n_swap  av. depth        min. depth  av. n_qubits      min. qubits    total time (s)  av. time (s)      min. time (s)
----------  -------------  -------------  -------------  ------------  --------------  -------------  ----------------  --------------  ---------------
line-graph  436 ± 0.0                436  723 ± 0.0               723  8 ± 0.0                     8              0.25  0.25 ± 0.0                 0.25
sabre       115 ± 1.59               115  609.25 ± 6.01           586  7.0 ± nan                   7              1.65  0.1 ± 0.02                 0.09
------------------------------------------------------------------------------------------------------------------------------------------------------
}
{\tiny
------------------------------------------------------------------------------------------------------------------------------------------------------
name = complete, size = 7, circuit_type = random, p = 700, repetitions = 16, optimization_level = 1

method      av. n_swaps      min. n_swap  av. depth        min. depth  av. n_qubits      min. qubits    total time (s)  av. time (s)      min. time (s)
----------  -------------  -------------  -------------  ------------  --------------  -------------  ----------------  --------------  ---------------
line-graph  478 ± 0.0                478  787 ± 0.0               787  8 ± 0.0                     8              0.26  0.26 ± 0.0                 0.26
sabre       118 ± 0.38               117  573.25 ± 1.66           567  7.0 ± nan                   7              2.24  0.14 ± 0.01                0.13
------------------------------------------------------------------------------------------------------------------------------------------------------
}
{\tiny
------------------------------------------------------------------------------------------------------------------------------------------------------
name = complete, size = 9, circuit_type = random, p = 900, repetitions = 16, optimization_level = 0

method      av. n_swaps      min. n_swap  av. depth       min. depth  av. n_qubits      min. qubits    total time (s)  av. time (s)      min. time (s)
----------  -------------  -------------  ------------  ------------  --------------  -------------  ----------------  --------------  ---------------
line-graph  593 ± 0.0                593  946 ± 0.0              946  10 ± 0.0                   10              0.43  0.43 ± 0.0                 0.43
sabre       151 ± 2.0                143  770.0 ± 6.84           730  9.0 ± nan                   9              1.6   0.1 ± nan                  0.1
------------------------------------------------------------------------------------------------------------------------------------------------------
}
{\tiny
------------------------------------------------------------------------------------------------------------------------------------------------------
name = complete, size = 9, circuit_type = random, p = 900, repetitions = 16, optimization_level = 1

method      av. n_swaps      min. n_swap  av. depth       min. depth  av. n_qubits      min. qubits    total time (s)  av. time (s)      min. time (s)
----------  -------------  -------------  ------------  ------------  --------------  -------------  ----------------  --------------  ---------------
line-graph  650 ± 0.0                650  1036 ± 0.0            1036  10 ± 0.0                   10              0.33  0.33 ± 0.0                 0.33
sabre       149 ± 0.59               146  726.5 ± 1.59           720  9.0 ± nan                   9              2.74  0.17 ± 0.02                0.16
------------------------------------------------------------------------------------------------------------------------------------------------------
}

    \end{Verbatim}

    \hypertarget{against-other-routing-methods}{%
\subsubsection{Against other routing
methods}\label{against-other-routing-methods}}

    Above, we only ran SABRE because it outperforms the other methods
available in Qiskit by default. The standard methods available are

    \begin{tcolorbox}[breakable, size=fbox, boxrule=1pt, pad at break*=1mm,colback=cellbackground, colframe=cellborder]
\prompt{In}{incolor}{42}{\boxspacing}
\begin{Verbatim}[commandchars=\\\{\}]
\PY{k+kn}{from} \PY{n+nn}{qiskit} \PY{k+kn}{import} \PY{n}{transpiler} 
\PY{n}{transpiler}\PY{o}{.}\PY{n}{preset\PYZus{}passmanagers}\PY{o}{.}\PY{n}{plugin}\PY{o}{.}\PY{n}{list\PYZus{}stage\PYZus{}plugins}\PY{p}{(}\PY{l+s+s1}{\PYZsq{}}\PY{l+s+s1}{routing}\PY{l+s+s1}{\PYZsq{}}\PY{p}{)}
\end{Verbatim}
\end{tcolorbox}

            \begin{tcolorbox}[breakable, size=fbox, boxrule=.5pt, pad at break*=1mm, opacityfill=0]
\prompt{Out}{outcolor}{42}{\boxspacing}
\begin{Verbatim}[commandchars=\\\{\}]
['basic', 'lookahead', 'none', 'sabre', 'stochastic']
\end{Verbatim}
\end{tcolorbox}
        
    \begin{tcolorbox}[breakable, size=fbox, boxrule=1pt, pad at break*=1mm,colback=cellbackground, colframe=cellborder]
\prompt{In}{incolor}{43}{\boxspacing}
\begin{Verbatim}[commandchars=\\\{\}]
\PY{k+kn}{import} \PY{n+nn}{line\PYZus{}graph\PYZus{}routing} \PY{k}{as} \PY{n+nn}{lgr}
\PY{k+kn}{import} \PY{n+nn}{pickle}

\PY{n}{settings}\PY{o}{=}\PY{p}{[}\PY{p}{]}
\PY{k}{for} \PY{n}{name} \PY{o+ow}{in} \PY{p}{[}\PY{l+s+s1}{\PYZsq{}}\PY{l+s+s1}{kagome}\PY{l+s+s1}{\PYZsq{}}\PY{p}{,}\PY{l+s+s1}{\PYZsq{}}\PY{l+s+s1}{shuriken}\PY{l+s+s1}{\PYZsq{}}\PY{p}{]}\PY{p}{:}
    \PY{k}{for} \PY{n}{side} \PY{o+ow}{in} \PY{p}{[}\PY{l+m+mi}{1}\PY{p}{]}\PY{p}{:}
        \PY{k}{for} \PY{n}{p} \PY{o+ow}{in} \PY{p}{[}\PY{l+m+mi}{1}\PY{p}{]}\PY{p}{:}
            \PY{k}{for} \PY{n}{optimization\PYZus{}level} \PY{o+ow}{in} \PY{p}{[}\PY{l+m+mi}{3}\PY{p}{]}\PY{p}{:}
                \PY{n}{setting}\PY{o}{=}\PY{p}{\PYZob{}}\PY{l+s+s1}{\PYZsq{}}\PY{l+s+s1}{name}\PY{l+s+s1}{\PYZsq{}}\PY{p}{:}\PY{n}{name}\PY{p}{,}
                         \PY{l+s+s1}{\PYZsq{}}\PY{l+s+s1}{size}\PY{l+s+s1}{\PYZsq{}}\PY{p}{:} \PY{p}{(}\PY{n}{side}\PY{p}{,}\PY{n}{side}\PY{p}{)}\PY{p}{,}
                         \PY{l+s+s1}{\PYZsq{}}\PY{l+s+s1}{circuit\PYZus{}type}\PY{l+s+s1}{\PYZsq{}}\PY{p}{:} \PY{l+s+s1}{\PYZsq{}}\PY{l+s+s1}{quantum\PYZus{}simulation}\PY{l+s+s1}{\PYZsq{}}\PY{p}{,}
                         \PY{l+s+s1}{\PYZsq{}}\PY{l+s+s1}{p}\PY{l+s+s1}{\PYZsq{}}\PY{p}{:} \PY{n}{p}\PY{p}{,}
                         \PY{l+s+s1}{\PYZsq{}}\PY{l+s+s1}{repetitions}\PY{l+s+s1}{\PYZsq{}} \PY{p}{:} \PY{l+m+mi}{16}\PY{p}{,}
                         \PY{l+s+s1}{\PYZsq{}}\PY{l+s+s1}{optimization\PYZus{}level}\PY{l+s+s1}{\PYZsq{}} \PY{p}{:} \PY{n}{optimization\PYZus{}level}\PY{p}{,}
                         \PY{l+s+s1}{\PYZsq{}}\PY{l+s+s1}{methods}\PY{l+s+s1}{\PYZsq{}} \PY{p}{:} \PY{p}{[}\PY{l+s+s1}{\PYZsq{}}\PY{l+s+s1}{basic}\PY{l+s+s1}{\PYZsq{}}\PY{p}{,} \PY{l+s+s1}{\PYZsq{}}\PY{l+s+s1}{lookahead}\PY{l+s+s1}{\PYZsq{}}\PY{p}{,} \PY{l+s+s1}{\PYZsq{}}\PY{l+s+s1}{sabre}\PY{l+s+s1}{\PYZsq{}}\PY{p}{,} \PY{l+s+s1}{\PYZsq{}}\PY{l+s+s1}{stochastic}\PY{l+s+s1}{\PYZsq{}}\PY{p}{]}
                    \PY{p}{\PYZcb{}}
                \PY{n}{settings}\PY{o}{.}\PY{n}{append}\PY{p}{(}\PY{n}{setting}\PY{p}{)}

\PY{c+c1}{\PYZsh{}\PYZsh{} Uncomment to rerun benchmarks}
\PY{c+c1}{\PYZsh{}results=[]}
\PY{c+c1}{\PYZsh{}for setting in settings:}
\PY{c+c1}{\PYZsh{}    result=lgr.benchmark(**setting)}
\PY{c+c1}{\PYZsh{}    results.append(result)}
\PY{c+c1}{\PYZsh{}    lgr.print\PYZus{}benchmark(result)}

\PY{c+c1}{\PYZsh{}with open(\PYZsq{}benchmark\PYZus{}data/other\PYZus{}methods\PYZus{}1x1.pkl\PYZsq{},\PYZsq{}wb\PYZsq{}) as f:}
\PY{c+c1}{\PYZsh{}    pickle.dump(results,f)}

\PY{k}{with} \PY{n+nb}{open}\PY{p}{(}\PY{l+s+s1}{\PYZsq{}}\PY{l+s+s1}{benchmark\PYZus{}data/other\PYZus{}methods\PYZus{}1x1.pkl}\PY{l+s+s1}{\PYZsq{}}\PY{p}{,}\PY{l+s+s1}{\PYZsq{}}\PY{l+s+s1}{rb}\PY{l+s+s1}{\PYZsq{}}\PY{p}{)} \PY{k}{as} \PY{n}{f}\PY{p}{:}
    \PY{n}{results}\PY{o}{=}\PY{n}{pickle}\PY{o}{.}\PY{n}{load}\PY{p}{(}\PY{n}{f}\PY{p}{)}

\PY{k}{for} \PY{n}{result} \PY{o+ow}{in} \PY{n}{results}\PY{p}{:}
    \PY{n}{lgr}\PY{o}{.}\PY{n}{print\PYZus{}benchmark}\PY{p}{(}\PY{n}{result}\PY{p}{)}
\end{Verbatim}
\end{tcolorbox}

    \begin{Verbatim}[commandchars=\\\{\}]
{\tiny
------------------------------------------------------------------------------------------------------------------------------------------------------
name = kagome, size = (1, 1), circuit_type = quantum_simulation, p = 1, repetitions = 16, optimization_level = 3

method      av. n_swaps      min. n_swap  av. depth      min. depth  av. n_qubits      min. qubits    total time (s)  av. time (s)      min. time (s)
----------  -------------  -------------  -----------  ------------  --------------  -------------  ----------------  --------------  ---------------
line-graph  12 ± 0.0                  12  7 ± 0.0                 7  12 ± 0.0                   12              0.08  0.08 ± 0.0                 0.08
basic       14 ± 0.0                  14  17.0 ± 0.0             17  8.0 ± 0.0                   8              0.05  0.05 ± 0.0                 0.05
lookahead   8 ± nan                    8  8.81 ± 0.19             8  8.0 ± nan                   8             19.55  1.22 ± 0.04                1.26
sabre       6 ± nan                    6  6.0 ± nan               6  8.0 ± nan                   8              0.49  0.03 ± 0.0                 0.03
stochastic  6 ± 0.62                   6  6.06 ± 0.16             6  8.0 ± nan                   8              0.62  0.04 ± 0.0                 0.03
------------------------------------------------------------------------------------------------------------------------------------------------------
}
{\tiny
------------------------------------------------------------------------------------------------------------------------------------------------------
name = shuriken, size = (1, 1), circuit_type = quantum_simulation, p = 1, repetitions = 16, optimization_level = 3

method      av. n_swaps      min. n_swap  av. depth      min. depth  av. n_qubits      min. qubits    total time (s)  av. time (s)      min. time (s)
----------  -------------  -------------  -----------  ------------  --------------  -------------  ----------------  --------------  ---------------
line-graph  8 ± 0.0                    8  9 ± 0.0                 9  8 ± 0.0                     8              0.05  0.05 ± 0.0                 0.05
basic       16 ± 0.0                  16  19.0 ± 0.0             19  8.0 ± 0.0                   8              0.04  0.04 ± 0.0                 0.04
lookahead   6 ± nan                    6  8.0 ± nan               8  8.0 ± nan                   8             14.22  0.89 ± 0.01                0.88
sabre       6 ± nan                    6  7.62 ± 0.22             7  8.0 ± nan                   8              0.5   0.03 ± 0.0                 0.04
stochastic  6 ± nan                    6  7.0 ± nan               7  8.0 ± nan                   8              0.58  0.04 ± 0.0                 0.04
------------------------------------------------------------------------------------------------------------------------------------------------------
}


    \end{Verbatim}

    The method `lookahead' takes an impractical amount of time, so we
exclude it when running benchmarks for larger/deeper circuits.

    \begin{tcolorbox}[breakable, size=fbox, boxrule=1pt, pad at break*=1mm,colback=cellbackground, colframe=cellborder]
\prompt{In}{incolor}{ }{\boxspacing}
\begin{Verbatim}[commandchars=\\\{\}]
\PY{k+kn}{import} \PY{n+nn}{line\PYZus{}graph\PYZus{}routing} \PY{k}{as} \PY{n+nn}{lgr}
\PY{k+kn}{import} \PY{n+nn}{pickle}

\PY{n}{settings}\PY{o}{=}\PY{p}{[}\PY{p}{]}
\PY{k}{for} \PY{n}{name} \PY{o+ow}{in} \PY{p}{[}\PY{l+s+s1}{\PYZsq{}}\PY{l+s+s1}{kagome}\PY{l+s+s1}{\PYZsq{}}\PY{p}{,}\PY{l+s+s1}{\PYZsq{}}\PY{l+s+s1}{shuriken}\PY{l+s+s1}{\PYZsq{}}\PY{p}{]}\PY{p}{:}
    \PY{k}{for} \PY{n}{side} \PY{o+ow}{in} \PY{p}{[}\PY{l+m+mi}{3}\PY{p}{]}\PY{p}{:}
        \PY{k}{for} \PY{n}{p} \PY{o+ow}{in} \PY{p}{[}\PY{l+m+mi}{3}\PY{p}{]}\PY{p}{:}
            \PY{k}{for} \PY{n}{optimization\PYZus{}level} \PY{o+ow}{in} \PY{p}{[}\PY{l+m+mi}{3}\PY{p}{]}\PY{p}{:}
                \PY{n}{setting}\PY{o}{=}\PY{p}{\PYZob{}}\PY{l+s+s1}{\PYZsq{}}\PY{l+s+s1}{name}\PY{l+s+s1}{\PYZsq{}}\PY{p}{:}\PY{n}{name}\PY{p}{,}
                         \PY{l+s+s1}{\PYZsq{}}\PY{l+s+s1}{size}\PY{l+s+s1}{\PYZsq{}}\PY{p}{:} \PY{p}{(}\PY{n}{side}\PY{p}{,}\PY{n}{side}\PY{p}{)}\PY{p}{,}
                         \PY{l+s+s1}{\PYZsq{}}\PY{l+s+s1}{circuit\PYZus{}type}\PY{l+s+s1}{\PYZsq{}}\PY{p}{:} \PY{l+s+s1}{\PYZsq{}}\PY{l+s+s1}{quantum\PYZus{}simulation}\PY{l+s+s1}{\PYZsq{}}\PY{p}{,}
                         \PY{l+s+s1}{\PYZsq{}}\PY{l+s+s1}{p}\PY{l+s+s1}{\PYZsq{}}\PY{p}{:} \PY{n}{p}\PY{p}{,}
                         \PY{l+s+s1}{\PYZsq{}}\PY{l+s+s1}{repetitions}\PY{l+s+s1}{\PYZsq{}} \PY{p}{:} \PY{l+m+mi}{16}\PY{p}{,}
                         \PY{l+s+s1}{\PYZsq{}}\PY{l+s+s1}{optimization\PYZus{}level}\PY{l+s+s1}{\PYZsq{}} \PY{p}{:} \PY{n}{optimization\PYZus{}level}\PY{p}{,}
                         \PY{l+s+s1}{\PYZsq{}}\PY{l+s+s1}{methods}\PY{l+s+s1}{\PYZsq{}} \PY{p}{:} \PY{p}{[}\PY{l+s+s1}{\PYZsq{}}\PY{l+s+s1}{basic}\PY{l+s+s1}{\PYZsq{}}\PY{p}{,} \PY{l+s+s1}{\PYZsq{}}\PY{l+s+s1}{sabre}\PY{l+s+s1}{\PYZsq{}}\PY{p}{,} \PY{l+s+s1}{\PYZsq{}}\PY{l+s+s1}{stochastic}\PY{l+s+s1}{\PYZsq{}}\PY{p}{]}
                    \PY{p}{\PYZcb{}}
                \PY{n}{settings}\PY{o}{.}\PY{n}{append}\PY{p}{(}\PY{n}{setting}\PY{p}{)}

\PY{c+c1}{\PYZsh{}\PYZsh{} Uncomment to rerun benchmarks}
\PY{c+c1}{\PYZsh{}results=[]}
\PY{c+c1}{\PYZsh{}for setting in settings:}
\PY{c+c1}{\PYZsh{}    result=lgr.benchmark(**setting)}
\PY{c+c1}{\PYZsh{}    results.append(result)}
\PY{c+c1}{\PYZsh{}    lgr.print\PYZus{}benchmark(result)}
\PY{c+c1}{\PYZsh{}}
\PY{c+c1}{\PYZsh{}with open(\PYZsq{}benchmark\PYZus{}data/other\PYZus{}methods\PYZus{}3x3.pkl\PYZsq{},\PYZsq{}wb\PYZsq{}) as f:}
\PY{c+c1}{\PYZsh{}    pickle.dump(results,f)}

\PY{k}{with} \PY{n+nb}{open}\PY{p}{(}\PY{l+s+s1}{\PYZsq{}}\PY{l+s+s1}{benchmark\PYZus{}data/other\PYZus{}methods\PYZus{}3x3.pkl}\PY{l+s+s1}{\PYZsq{}}\PY{p}{,}\PY{l+s+s1}{\PYZsq{}}\PY{l+s+s1}{rb}\PY{l+s+s1}{\PYZsq{}}\PY{p}{)} \PY{k}{as} \PY{n}{f}\PY{p}{:}
    \PY{n}{results}\PY{o}{=}\PY{n}{pickle}\PY{o}{.}\PY{n}{load}\PY{p}{(}\PY{n}{f}\PY{p}{)}

\PY{k}{for} \PY{n}{result} \PY{o+ow}{in} \PY{n}{results}\PY{p}{:}
    \PY{n}{lgr}\PY{o}{.}\PY{n}{print\PYZus{}benchmark}\PY{p}{(}\PY{n}{result}\PY{p}{)}
  \end{Verbatim}
\end{tcolorbox}

\begin{Verbatim}[commandchars=\\\{\}]
{\tiny
------------------------------------------------------------------------------------------------------------------------------------------------------
name = kagome, size = (3, 3), circuit_type = quantum_simulation, p = 3, repetitions = 16, optimization_level = 3

method      av. n_swaps      min. n_swap  av. depth       min. depth  av. n_qubits      min. qubits    total time (s)  av. time (s)      min. time (s)
----------  -------------  -------------  ------------  ------------  --------------  -------------  ----------------  --------------  ---------------
line-graph  336 ± 0.0                336  43 ± 0.0                43  68 ± 0.0                   68              1.33  1.33 ± 0.0                 1.33
basic       475 ± 0.0                475  287.0 ± 0.0            287  45.0 ± 0.0                 45              1.21  1.21 ± 0.0                 1.21
sabre       215 ± 3.56               203  63.25 ± 3.78            50  43.94 ± 0.97               45             12.37  0.77 ± 0.04                0.93
stochastic  410 ± 14.31              428  85.88 ± 4.12            71  46.44 ± 1.5                48             28.09  1.76 ± 0.07                1.46
------------------------------------------------------------------------------------------------------------------------------------------------------
}
{\tiny
------------------------------------------------------------------------------------------------------------------------------------------------------
name = shuriken, size = (3, 3), circuit_type = quantum_simulation, p = 3, repetitions = 16, optimization_level = 3

method      av. n_swaps      min. n_swap  av. depth       min. depth  av. n_qubits      min. qubits    total time (s)  av. time (s)      min. time (s)
----------  -------------  -------------  ------------  ------------  --------------  -------------  ----------------  --------------  ---------------
line-graph  390 ± 0.0                390  34 ± 0.0                34  84 ± 0.0                   84              1.81  1.81 ± 0.0                 1.81
basic       1195 ± 0.0              1195  491.0 ± 0.0            491  71.0 ± 0.0                 71              2.14  2.14 ± 0.0                 2.14
sabre       380 ± 10.03              344  78.19 ± 4.42            66  64.19 ± 0.94               67             24.94  1.56 ± 0.12                2.31
stochastic  813 ± 32.04              609  104.0 ± 7.88            76  66.56 ± 1.44               71             62.25  3.89 ± 0.09                3.76
------------------------------------------------------------------------------------------------------------------------------------------------------
}
  \end{Verbatim}
    We see SABRE outperforms the other methods available by default in
Qiskit, but not line-graph routing for the circuits considered.

    \hypertarget{wall-clock-time-of-line-graph-routing}{%
\subsubsection{Wall-clock time of line-graph
routing}\label{wall-clock-time-of-line-graph-routing}}

    Create a random graph, construct the line graph, create a circuit on the
line graph, and put this circuit into line-graph routing

    \begin{tcolorbox}[breakable, size=fbox, boxrule=1pt, pad at break*=1mm,colback=cellbackground, colframe=cellborder]
\prompt{In}{incolor}{44}{\boxspacing}
\begin{Verbatim}[commandchars=\\\{\}]
\PY{k+kn}{from} \PY{n+nn}{time} \PY{k+kn}{import} \PY{n}{time}
\PY{n}{side}\PY{o}{=}\PY{l+m+mi}{25}
\PY{n}{lg} \PY{o}{=} \PY{n}{lgr}\PY{o}{.}\PY{n}{kagome}\PY{p}{(}\PY{n}{side}\PY{p}{,} \PY{n}{side}\PY{p}{)}
\PY{n+nb}{print}\PY{p}{(}\PY{l+s+s1}{\PYZsq{}}\PY{l+s+s1}{number of nodes =}\PY{l+s+s1}{\PYZsq{}}\PY{p}{,}\PY{n}{lg}\PY{o}{.}\PY{n}{number\PYZus{}of\PYZus{}nodes}\PY{p}{(}\PY{p}{)}\PY{p}{)}
\PY{n}{La}\PY{o}{=}\PY{l+m+mi}{10}\PY{o}{*}\PY{o}{*}\PY{l+m+mi}{5}
\PY{n+nb}{print}\PY{p}{(}\PY{l+s+s1}{\PYZsq{}}\PY{l+s+s1}{number of gates =}\PY{l+s+s1}{\PYZsq{}}\PY{p}{,}\PY{n}{La}\PY{p}{)}
\PY{n}{qc} \PY{o}{=} \PY{n}{lgr}\PY{o}{.}\PY{n}{random\PYZus{}circuit}\PY{p}{(}\PY{n}{lg}\PY{p}{,} \PY{n}{La}\PY{p}{)}
\PY{n}{begin}\PY{o}{=}\PY{n}{time}\PY{p}{(}\PY{p}{)}
\PY{n}{qc} \PY{o}{=} \PY{n}{lgr}\PY{o}{.}\PY{n}{line\PYZus{}graph\PYZus{}route}\PY{p}{(}\PY{n}{qc}\PY{p}{)}
\PY{n}{end}\PY{o}{=}\PY{n}{time}\PY{p}{(}\PY{p}{)}
\PY{n+nb}{print}\PY{p}{(}\PY{l+s+s1}{\PYZsq{}}\PY{l+s+s1}{wall clock time =}\PY{l+s+s1}{\PYZsq{}}\PY{p}{,}\PY{n}{end}\PY{o}{\PYZhy{}}\PY{n}{begin}\PY{p}{,}\PY{l+s+s1}{\PYZsq{}}\PY{l+s+s1}{(s)}\PY{l+s+s1}{\PYZsq{}}\PY{p}{)}
\end{Verbatim}
\end{tcolorbox}

    \begin{Verbatim}[commandchars=\\\{\}]
number of nodes = 1976
number of gates = 100000
wall clock time = 25.544434070587158 (s)
    \end{Verbatim}
